\documentclass[a4paper,fleqn]{article}
\usepackage{amsmath}
\usepackage{latexsym}
\usepackage{amssymb}
\usepackage{amsfonts}
\usepackage{mathrsfs}
\usepackage{graphicx}
\usepackage[all]{xy}
\usepackage{amsthm}
\usepackage{enumerate}
\usepackage{url}
\usepackage{float}

\DeclareMathOperator{\Fin}{\mathit{Fin}}

\newcommand{\circtop}[1]{\overset{\text{\tiny$\circ$}}{#1}}
\newcommand{\mon}[2]{\mathrm{mon}_{#1}\left(#2\right)}

\newcommand*{\FN}{\mathit{FN}}
\newcommand*{\FQ}{\mathit{FQ}}
\newcommand*{\BQ}{\mathit{BQ}}

\newcommand{\choice}[2]{{\mathrm{Ch}_{{#1}}^{{#2}}}}
\let\uhr\upharpoonright
\renewcommand*{\upharpoonright}{\hspace{-.07cm}\uhr\hspace{-.07cm}}
\newtheorem{thm}{Theorem}
\newtheorem{prop}[thm]{Proposition}

\newtheorem{lemma}[thm]{Lemma}
\theoremstyle{remark}

\theoremstyle{definition}
\newtheorem{defn}[thm]{Definition}
\newtheorem{ex}{Example}
\newtheoremstyle{axiom1}
{3pt}
{3pt}
{\itshape}
{}
{\bfseries\itshape}
{}
{.5em}
{}
\newtheoremstyle{axiom2}
{3pt}
{3pt}
{\itshape}
{}
{\bfseries\itshape}
{\\}
{.5em}
{}
\theoremstyle{axiom1}
\title{Extending Games beyond the Finite Horizon\thanks{The authors thank Akihiko Matsui for his advices and criticims. This paper would not exist at all if he had not suggested that the authors' framework could be applied to the Centipede games.}}
\author{Kiri Sakahara\thanks{Yokohama National University and Kanagawa University, Kanagawa, Japan.} \and Takashi Sato\thanks{Toyo University, Tokyo, Japan.}}
\date{}
\begin{document}

\maketitle

\begin{abstract}
This paper argues that the finite horizon paradox, where game theory contradicts intuition, stems from the limitations of standard number systems in modelling the cognitive perception of infinity.
To address this issue, we propose a new framework based on Alternative Set Theory (AST).
This framework represents different cognitive perspectives on a long history of events using distinct topologies.
These topologies define an indiscernibility equivalence that formally treats huge, indistinguishable quantities as equivalent.
This offers criterion-dependent resolutions to long-standing paradoxes, such as Selten's chain store paradox and Rosenthal's centipede game.
Our framework reveals new intuitive subgame perfect equilibria, the characteristics of which depend on the chosen temporal perspective and payoff evaluation.
Ultimately, by grounding its mathematical foundation in different modes of human cognition, our work expands the explanatory power of game theory for long-horizon scenarios.
\end{abstract}

\section{Introduction}
Among the problems regarding the discrepancy between the outcomes predicted by the game theory and those consistent with our intuition, Rubinstein \cite{br} draws attention to the importance of the ``finite horizon paradox.''
The most notable examples of the paradox, as Rubinstein mentions, are ``the finitely repeated Prisoner's Dilemma, Rosenthal's centipede game, and Selten's chain store paradox.''\footnote{l.2-4 p.165 of Rubinstein \cite{br}.}
While progress in the field of repeated games has contributed significantly to organising the first problem from a particular point of view, insufficient progress has been made on the last two.
The present paper aims to address the problem inherent in all of these three by focusing on the question of how numbers are constructed in the traditional set theory.

The standard number system, based on the Zermelo-Fraenkel set theory (ZF for short), has been adopted without much scrutiny of its suitability for dealing with numerical aspects of the phenomenal world. 
However, when it comes to deal with subjective reality, the system has some weaknesses, especially when it comes to phenomena involving infinity.

Suppose, for example, there is a collection of a billion one-dollar bills in bulk.
It would be obvious to almost anyone that it is indistinguishable from the other, which consists of a billion plus one.
The two may appear to be infinitely many, and therefore indiscernible.
The same is true of any pair of extremely small numbers, which are also indiscernible to us.
The authors believe that exactly the same mechanism that makes two numbers seem indiscernible may lie behind the finite horizon paradox.

To deal with these phenomena, therefore, it is not appropriate to adopt ZF as the basis of a number system, since it has no mechanism to consider the two as indiscernible and thus cannot adequately deal with these problems. In order to grasp the mechanism behind them, it is necessary to introduce a number system that can adequately represent such phenomena.

The present paper adopts Alternative Set Theory (AST for short) as a new basis for a slightly different number system.
It allows to construct a system in which two extremely huge or small numbers are regarded as indiscernible from each other.
A general framework is also introduced which allows all three games mentioned above to be seen as special cases.
Within this framework, the paradoxes inherent in the games mentioned at the beginning are resolved in their proper context.

\section{A System of Numbers}

Let us start with the construction of a number system according to  Vop\v{e}nka \cite{ast}.
The class of \textit{natural numbers} $N$ is defined as follows:
\[
N\ =\ \left\{x\ ;\ 
\begin{matrix}
\left(\forall y\in x\right) \left(y\subseteq x\right)\\
\wedge\left(\forall y,z\in x \right) \left(y\in z \vee y=z \vee z\in y\right)
\end{matrix}
\right\},
\]
while the class of \textit{finite natural numbers} $\FN$ consists of the numbers represented by finite sets
\[
\FN \ =\ \left\{x\in N\ ;\ 
\Fin(x)\right\}
\]
where $\Fin(x)$ means that every subclass of $x$ is a set.
Note that $N\setminus\FN\ne\emptyset$, or $\FN$ is a proper subclass of $N$, since there are huge natural numbers.
All these huge natural numbers include $\FN$ as a subclass, and so by definition they cannot be finite natural numbers.

The class of all integers $Z$ and that of all rational numbers are defined respectively as:
\[
Z\ =\ N\cup \left\{ -a;\,
 a\in N
\right\}
\qquad\text{and}\qquad
Q\ =\ \left\{
\frac{x}{y}\ ;\ x,\,y\in Z \wedge y\ne 0
\right\}.
\]
$\BQ\subseteq Q$ denotes the class of \textit{bounded rational numbers} and $\FQ\subseteq BQ$ the class of \textit{finite rational numbers}, i.e. 
\[
\BQ=\{x\in Q\ ;\ \left(\exists i\in\FN\right)
  \left(|x|\leq i\right)
\}
\quad\text{and}\quad
\FQ\ =\ \left\{
\frac{x}{y}\ ;\ 
x,y\in\FN \wedge y\ne 0
\right\}.
\]

Real numbers are defined in AST as an equivalence class of bounded rational numbers.
The reason for this construction is in the human inability to distinguish between two rational numbers that are close to each other.
This idea is captured by the \textit{indiscernibility equivalence}, $\doteq$, on the class $Q$ of all rational numbers.
One of the definitions of the indiscernibility equivalence $\doteq$ is given as: 
\[
p\doteq q\quad \equiv\quad
 \begin{pmatrix}
  \left( \exists k \right)
  \left( \forall i>0\right)
  \left(|p|<k \wedge |p-q|<\frac{1}{i}\right)\\[.5em]
  \vee 
  \left(\forall k\right)
  \left(
  \left( p>k \wedge q>k\right)  \vee 
  \left(p<-k \wedge q<-k\right)
  \right)
 \end{pmatrix}
\]
where the letters $i,j,k$ denote finite natural numbers, i.e. $i,j,k\in\FN$ for notational convenience.
For any $q\in Q$ the notation $\mon{}{q}=\{s\in Q\,;\, s\doteq q\}$, is said to be a \textit{monad} of $q$, representing the class of all rational numbers that are indiscernible from $q$.
The real number $a$ is denoted as a monad $\mon{}{q}$ of some rational number $q$.
Two limit cases are denoted as:
\[
\infty\ =\ 
\{q\in Q\,;\, 
 \left(\forall i\right)
 \left(q>i\right)
\}\quad\text{ and }\quad
-\infty\ =\ 
\{q\in Q\,;\,
 \left(\forall i\right)
 \left(q<-i\right)
\}.
\]
The class of all real numbers $R$ and that of plus and minus infinity $R^+$ are defined as:
\[
R\ \equiv\ \left\{\mon{}{x}\ ;\ 
x\in \BQ
\right\}\ =\ \BQ/{\doteq}
\quad\text{ and }\quad
R^+\ \equiv\ R\cup\{-\infty,\infty\}.
\]
A \textit{real continuum} is denoted by $\mathscr{R}=\langle Q,\doteq\rangle$, where a \textit{continuum} is a pair of classes $\mathscr{C}=\left\langle C,\doteq_C\right\rangle$, where a set-theoretically definable class $C$ is called as a support of $\mathscr{C}$.

\section{Extensive Games with Perfect Information}

Let us review the basic concepts of extensive games. 
In order to summarise these concepts, Osborne and Rubinstein \cite{OR} is exclusively consulted.
Almost all definitions are based on the book, with a few exceptions necessary to conform to the notation in our framework.

\begin{defn}[Definition 89.1 of \cite{OR}]\label{def_ext}
  An \textit{extensive game with perfect information} has the following components
  \begin{itemize}
  \item A class $I$ of \textit{players}.
  \item A class $H$ of sequences of \textit{actions}, denoted by $(a_k)_{k=1}^\kappa=\{\langle k,a_k\rangle\,;\,k\in\{1,\ldots,\kappa\}\}$,\footnote{The notation of sets here follows the style of Kunen \cite{Kunen}.} (where $H$ is finite) that satisfies the following.\footnote{In Osborne and Rubinstein \cite{OR}, the number of periods in which each action occurs is indicated by superscripts.  
  This distinguishes them from the numbers representing each player, which are described by subscripts.
  In this article, however, both are described by subscripts in order to distinguish them from the numbers representing the  periods of the whole histories.
  The definition of whole histories will be provided later.
  }.
    \begin{itemize}
    \item The empty sequence $\emptyset$, so called \textit{initial history}, is a member of $H$.
    \item If $(a_k)_{k=1}^\kappa\in H$ (where $\kappa\in \FN$) and $\lambda\in \kappa$ then  $(a_k)_{k=1}^\lambda\in H$.
    \end{itemize}
    Each member of $H$ is a \textit{history}; each component of a history is an ordered pair of a period $k$ and an \textit{action} taken at $k$.
    A history $(a_k)_{k=1}^\kappa \in H$ is \textit{terminal} if there is no $a_{\kappa+1}$ such that $(a_k)_{k=1}^{\kappa+1}\in H$.
    The set of terminal histories is denoted $Z$.
  \item A function $P$ that assigns to each nonterminal history (each member of $H\setminus Z$) a member of $I$
    ($P$ is the \textit{player function}, $P(h)$ being the player who takes an action after the history $h$).
  \item For each player $i\in I$ a preference relation $\succsim_i$ on $Z\cup\{\emptyset\}$\footnote{A history preferred to $\emptyset$ can be interpreted as being better than nothing.} (the \textit{preference relation} of player $i$).
  \end{itemize}
\end{defn}


The class of actions from which the player $P(h)$ chooses after a history $h\in H$ of extended games is denoted by
\[
A(h)\ =\ \{a\,;\, h^\frown (a)\in H\}.
\]
where $h^\frown(a)$ concatenates the action $a$ to the sequence of actions $h=(a_1,\ldots, a_{k})$ where $k<\kappa$ is satisfied, that is, $(a_1,\ldots,a_{k},a)$.
It is also used to describe the history $h$ followed by $h'=(b_1,\ldots,b_{k'})$ where $k+k'\leq \kappa$ is satisfied, as $h^\frown h'$ abbreviated $(a_1,\ldots,a_k,b_1,\ldots,b_{k'})$.

The quadruple $\Gamma=\langle I,H,P,(\succsim_i)\rangle$ is said to be an \textit{extensive game form with perfect information}.

The strategy of player $i$ is given as follows. 

\begin{defn}[Definition 92.1 of \cite{OR}]
  A \textit{strategy $s_i$ of player} $i\in I$ in an extensive game with perfect information $\Gamma=\langle I,H,P,(\succsim_i)\rangle$ is a function that assigns an action in $A(h)$ to each nonterminal history $h\in H\setminus Z$ for which $P(h)=i$.
\end{defn}

For each strategy profile $s=(s_i)_{i\in I}$ in the extensive game $\langle I,H,P,(\succsim_i)\rangle$, the \textit{outcome} $O(s)$ of $s$ is defined as the terminal history that results when each player $i\in I$ follows the strategy $s_i$.
Briefly, $O(s)$ denotes the terminal history $(a_1,\ldots,a_\kappa)\in Z$ that satisfies $s_{P(a_1,\ldots,a_k)}((a_1,\ldots,a_k))=a_{k+1}$  for each $k\in \{1,\ldots,\kappa-1\}$.

Two fundamental equilibrium concepts are described by this function.
The first one is Nash equilibrium.

\begin{defn}[Definition 93.1 of \cite{OR}]\label{Nash}
  A \textit{Nash equilibrium of an extensive game with perfect information} $\langle I, H, P, (\succsim_i)\rangle$ is a strategy profile $s^*$ such that for each player $i\in I$ the following condition is satisfied
  \[
  O(s^*_{-i}, s^*_i) \succsim_i   O(s^*_{-i}, s_i) \text{ for each strategy } s_i \text{ of player }i.
  \]
\end{defn}

To introduce the second one, subgame perfect equilibrium, it is necessary to introduce \textit{subgames} in advance.

\begin{defn}[Definition 97.1 of \cite{OR}]
  The \textit{subgame of the extensive game with perfect information} $\Gamma=\langle I,H,P,(\succsim_i)\rangle$ \textit{that follows the history $h$} is the extensive game $\Gamma(h)=\langle I,H|_h,P|_h,({\succsim_i}|_h)\rangle$, where $H|_h$ is the set of sequences $h'$ of actions for which $h^\frown h'\in H$, $P|_h$ is defined by $P|_h(h')=P(h^\frown h')$ for each $h'\in Z|_h$, and ${\succsim_i}|_h$ is defined by $h'{\succsim_i}|_h h''$ if and only if $h^\frown h'\succsim_i h^\frown h''$.
\end{defn}

Given a strategy $s_i$ of player $i\in I$ and a nonterminal history $h\in H\setminus Z$ in the extensive game $\Gamma$, $s_i|_h$ denotes the strategy $s_i|_h(h')=s_i(h^\frown h')$ for each $h'\in H|_h$.
The outcome function of $\Gamma(h)$ is denoted as $O_h$.

Finally, the concept of subgame perfect equilibrium is given as follows.

\begin{defn}[Definition 97.2 of \cite{OR}]\label{SPNE}
  A \textit{subgame perfect equilibrium of an extensive game with perfect information} $\Gamma=\langle I,H,P,(\succsim_i)\rangle$ is a strategy profile $s^*$ such that for every player $i\in I$ and every nonterminal history $h\in H\setminus Z$ for which $P(h)=i$ the following condition is satisfied
  \[
  O_h(s^*_{-i}|_h, s^*_i|_h)\ {\succsim_i}|_h \ O_h(s^*_{-i}|_h, s_i)
  \]
  for every strategy $s_i$ of player $i$ in the subgame $\Gamma(h)$.
\end{defn}

\section{Generalised Repeated Games}

Given a series of classes of connected terminal histories $h_1,\ldots,h_\tau\in C\subseteq Z$ that are connected to the next period,
a sequence of histories $\mathbf{h}=(h_1,\ldots,h_\tau)$ denotes a \textit{$\tau$-whole history} and ${C}^\tau=\underset{\tau \text{ times}}{\underbrace{C\times\cdots\times C}}$, where ${C}^0=\emptyset$ denotes a class of all $\tau$-whole histories.
The $\tau$-whole history that repeats a terminal history $h$ $\tau$ times is also denoted by $h^\tau$, where $h^0=\emptyset$.
As in the extensive game setting, $\mathbf{h}^\frown\mathbf{j}$ also abbreviates the whole history
$(h_1\ldots,h_{t}^\frown j_0,\ldots,j_{t'})$ if $h_t$ is a non-terminal history and $j_0$ is a connected terminal history of the subgame $\Gamma(h_t)$, i.e. $j_0\in C|_{h_t}$, where $\mathbf{h}=(h_1,\ldots,h_t)$ and $\mathbf{j}=(j_0,\ldots,j_{t'})$ satisfying $t+t'\leq \tau$.
It also abbreviates $(h_1,\ldots,h_{t},j_1,\ldots,j_{t'})$ if $h_t$ is a connected terminal history, i.e. $h_t\in C$, where $\mathbf{h}=(h_1,\ldots,h_t)$ and $\mathbf{j}=(j_1,\ldots,j_{t'})$ satisfying $t+t'\leq \tau$.

\begin{defn}
  Let $\Gamma=\langle I,H,P,(\succsim_i)\rangle$ be an extensive game, called a \textit{constituent game},
  and assume that $I$ is divided into two disjoint subclasses: the class of \textit{core players} $I_{\text{cor}}$ and the class of homogeneous \textit{outside players} $I_{\text{out}}$.
  Then, a \textit{$\tau$-repeated game of $\Gamma$} is an extensive game with perfect information $\mathbf{\Gamma}^\tau=\langle \mathbf{I}^\tau, \mathbf{H}^\tau, \mathbf{P}^\tau,(\succsim^\tau_i)\rangle$ where
  \begin{itemize}
  \item $\mathbf{I}^\tau= I_{\text{cor}}\cup \bigsqcup_{t\in\{1,\ldots,\tau\}} I_{{\text{out}}}$  where $I_{\text{cor}}\cup I_{\text{out}}=I$, $I_{\text{cor}}\cap I_{\text{out}}=\emptyset$ and $I_{\text{cor}}\ne\emptyset$
  \item $\mathbf{H}^\tau=\bigcup_{t\in \{1,\ldots,\tau-1\}}\left(C^t\times H\right)$ where $C\subseteq Z$ is a class of all \textit{connected terminal histories} of $\Gamma$ that are connected to a next period 
  \item $\mathbf{P}^\tau(\mathbf{h})=P({h}_k)$ for $\mathbf{h}=(h_1,\ldots,h_k)\in \mathbf{H}^\tau$
  \item $\succsim_i^\tau$ is a preference relation on $\mathbf{Z}^\tau=\bigcup_{t\in \{1,\ldots,\tau\}}Z^t$, denoting a class of all \textit{terminal histories} ($\mathbf{C}^\tau=\bigcup_{t\in \{1,\ldots,\tau-1\}}C^t$ denotes all connected terminals, which is a subclass of $\mathbf{Z}^\tau$).  \end{itemize}
\end{defn}

\begin{defn}
  A preference relation $\succsim_i^\tau$ satisfies \textit{weak separability} if the relation $\mathbf{h}'^\frown ({j})^\frown \mathbf{h}''' \succsim_i^\tau \mathbf{h}'^\frown ({j}')^\frown \mathbf{h}'''$ holds for any pair of whole histories $\mathbf{h}'\in{C}^{k}$ and $\mathbf{h}'''\in{C}^{\tau-k-1}$ where $k\in\{1,\ldots,\tau-1\}$, and a pair of histories ${j}, {j}'\in {C}$ satisfying $\mathbf{h}'^\frown({j})^\frown\mathbf{h}''',\ \mathbf{h}'^\frown({j}')^\frown\mathbf{h}'''\in \mathbf{Z}^\tau$ and ${j}\succsim_i {j}'$.
\end{defn}
  
  It may seem that strengthening the weak separability to strict separability is harmless.
But it is not, especially if a preference relation $\succsim_i$ is \textit{compact}, which means that for any infinite set $u$ there are $x,y\in u$ satisfying $x\ne y$ and $x\sim_i y$.
The next proposition confirms this problem.

\begin{prop}\label{ss}
  If a preference relation $\succsim_i^\tau$ on $\mathbf{Z}^\tau$ is compact, 
  it cannot satisfy strict separability, which replaces $\succsim_i$ with strict preferences $\succ_i$.\end{prop}
\begin{proof}
  Suppose that $h\succ_i j$ for some $h,j\in C$ and ${h}^\tau\succ_i^\tau {{h}^{\tau-1}}^\frown( {j})\succ_i^\tau\cdots\succ_i^\tau ({h})^\frown {j}^{\tau-1}\succ_i^\tau {j}^{\tau}$ holds for $\tau\in N\setminus\FN$.
  Since $\sim_i$ is compact, for any set $\mathbf{X}\subseteq \mathbf{Z}^\tau$ consisting of a huge number of whole histories and for any whole history $\mathbf{h}'\in \mathbf{X}$, there exists $\mathbf{h}''\ne \mathbf{h}'$ which is indifferent from $\mathbf{h}'$. 
  This implies that there are $\beta,\gamma\in \tau\setminus\FN$ satisfying ${h^{\tau-\beta}}^\frown j^{\beta}\sim_i^\tau {h^{\tau-\gamma}}^\frown j^\gamma$.
  This is a contradiction.
\end{proof}

In AST, it is assumed that given an infinite number of alternatives, some of them will always be considered the same.
Compactness represents this property.
It also plays a very important role in representing payoffs.

The strategy of player $i$ in a $\tau$-repeated game of ${\Gamma}$ is given by as follows.

\begin{defn}
  A \textit{strategy $\mathbf{s}_i$ of player} $i\in I$ in a $\tau$-repeated game of $\Gamma$ with perfect information $\mathbf{\Gamma}^\tau$ is a function that assigns an action in ${A}({h}_k)$ to each nonterminal whole history $\mathbf{h}=(h_1,\ldots, h_k)\in\mathbf{H}^\tau\setminus\mathbf{Z}^\tau$ which satisfies ${P}({h}_k)=i$ and an action in ${A}(\emptyset)$ to each connected terminal history $\mathbf{h}\in\mathbf{C}^\tau$ which satisfies ${P}(\emptyset)=i$.
\end{defn}

Given that $s$ is a strategy of the game $\Gamma$, ${s}^\tau$ denotes the strategy that decides actions according to $s$ for every $t$-th game of $t\in \{1,\ldots,\tau\}$.

The outcome $\mathbf{O}(\mathbf{s})$ of $\mathbf{\Gamma}^\tau$ is given by the unconnected terminal whole history
$(h_1,\ldots,h_{\tau})\in\mathbf{Z}^\tau\setminus\mathbf{C}^\tau$ which satisfies $\mathbf{s}_{\mathbf{P}^t((h_1,\ldots,h_{t}))}(h_1,\ldots,h_t)=h_{t+1}(1)$ for each $t\in\{1,\ldots,\tau-1\}$ and $\mathbf{s}_{\mathbf{P}^\tau((h_1,\ldots,h_t\, \upharpoonright\,  \ell))}((h_1,\ldots,h_t\upharpoonright\ell))=h_t(\ell)$\footnote{This is the usual notation in set theory, where a sequence $h_t=(a_1,\ldots,a_\kappa)$ for each $t\in\{1,\dots,\tau\}$ is denoted by a function from $\{1,\ldots,\kappa\}$ to a class of actions and the $\ell$-th element $a_\ell$ of $h_t$ is given by $h_t(\ell)$.
It is also common to represent the function with its domain restricted to $\{1,\ldots,\ell-1\}$ as $h_t\upharpoonright\ell$ and the cardinality of $h_t$ as $|h_t|$. See Kunen \cite{Kunen} for details.} for each $\ell\leq|h_t|$ and $t\in\{1,\ldots,\tau\}$.

  Nash equilibrium can be defined in the same way as in Definition \ref{Nash} by replacing $O(s)$ by $\mathbf{O(s)}$.
  The subgame of $\tau$-repeated games is given as follows.

\begin{defn}
  The \textit{subgame of the $\tau$-repeated game with perfect information} $\mathbf{\Gamma}^\tau=\langle \mathbf{I}^\tau,\mathbf{H}^\tau,\mathbf{P}^\tau,(\succsim_i^\tau)\rangle$ \textit{following the whole history $\mathbf{h}$} is the $\tau-|\mathbf{h}|$-repeated (or $\tau-|\mathbf{h}|+1$-repeated) 
  game $\mathbf{\Gamma}^\tau\mathbf{(h)}=\langle \mathbf{I}^\tau,\mathbf{H}^\tau|_{\mathbf{h}},\mathbf{P}^\tau|_\mathbf{h},({\succsim_i^\tau}|_\mathbf{h})\rangle$, where $\mathbf{H}^\tau|_\mathbf{h}$ is the set of whole histories $\mathbf{h}'$ that satisfy $\mathbf{h}^\frown \mathbf{h}'\in \mathbf{H}^\tau$, $\mathbf{P}^\tau|_\mathbf{h}$ is defined by $\mathbf{P}^\tau|_\mathbf{h}(\mathbf{h}')=\mathbf{P}^\tau(\mathbf{h}^\frown \mathbf{h}')$ for each $\mathbf{h}'\in \mathbf{H}^\tau|_\mathbf{h}$, and ${\succsim_i^\tau}|_\mathbf{h}$ is defined by $\mathbf{h}'{\succsim_i^\tau}|_\mathbf{h} \mathbf{h}''$ if and only if $\mathbf{h}^\frown \mathbf{h}'\succsim_i^\tau \mathbf{h}^\frown \mathbf{h}''$.
\end{defn}

Given a strategy $\mathbf{s}_i$ of player $i\in I$ and a nonterminal whole history $\mathbf{h}\in \mathbf{H}^\tau\setminus \mathbf{Z}^\tau\cup \mathbf{C}^\tau$, $\mathbf{s}_i|_{\mathbf{h}}$ denotes the strategy $\mathbf{s}_i|_{\mathbf{h}}(\mathbf{h}')=\mathbf{s}_i(\mathbf{h}^\frown \mathbf{h}')$ for each $\mathbf{h}'\in \mathbf{H}^\tau|_\mathbf{h}$.
The outcome function  of $\mathbf{\Gamma}(\mathbf{h})$ is also defined as $\mathbf{O}_{\mathbf{h}}$.

The subgame perfect equilibrium is also defined as in Definition \ref{SPNE} by replacing $O(s)$ by $\mathbf{O(s)}$.
However, to guarantee that the collection $s^{*\tau}$ of strategies that repeat a subgame perfect strategy profile $s^*$ of $\Gamma$ $\tau$ times is also that of the $\tau$-repeated game too, it is necessary to introduce the following condition.

\begin{defn}
  A preference relation $\succsim_i^\tau$ satisfies \textit{huge transitivity} if $\kappa$ is huge and $(\mathbf{h}_k)_{k=1}^\kappa$ is a chain of $\succsim_i^\tau$, i.e.  $\mathbf{h}_{k+1}\succsim_i^\tau \mathbf{h}_{k}$ for all $k\in\{1,\ldots, \kappa-1\}$, then for any $\ell,m\in \{1,\ldots,\kappa\}$ and $\ell\geq m$, the relation $\mathbf{h}_\ell\succsim_i^\tau \mathbf{h}_m$ holds.
\end{defn}

It may seem that huge transitivity is trivially satisfied.
But it is not if the numbers are given according to the AST setting.
For example, for any huge $\tau\in N\setminus \FN$, $\frac{t}{\tau}\doteq \frac{t+1}{\tau}$ holds for all $t\in\{1,\ldots,\tau-1\}$.
However, $\frac{0}{\tau}$ and $\frac{\tau}{\tau}$ are discernible, since $\frac{0}{\tau}=0\not\doteq 1=\frac{\tau}{\tau}$.
The huge transitivity condition prevents such a situation. 

It is also worth noting that the effectiveness of backward induction cannot be guaranteed without huge transitivity.
In fact, it may happen that $j^\tau$ is strictly preferable to $h^\tau$ even if the history $h$ is preferable to $j$ in its constituent game.
Therefore, it is necessary to assume the huge transitivity condition when using the backward induction procedure.
\\

Let us now look at some examples to see how these concepts and conditions work.

\begin{ex}[Chain Store Game]\label{cs}
A chain store game, originating from Selten \cite{chainstore}, consists of $T+1$ players: a chain store (player CS) and $T$ local stores (player $t$ for each $t\in \{1,\ldots,T\}$) in different cities.
The chain store also has its branches in all $T$ cities.
Each local store plans whether to open a second store in its own city, one by one from 1 to $T$, and the chain store is forced to choose whether to react `cooperatively' (C) or `aggressively' (A) if the local store decides to open a second store.

This situation can be captured by a constituent game with only two players: the chain store and a local store.
The core player is CS, while player 1 is the outside player.
The game tree is shown in Figure \ref{con_cs}.

\begin{figure}[htbp]
  \[
  \centerline{
  \xymatrix@=10pt@M=-.10pt{
    & &&
    \raise0.2ex\hbox{\textcircled{\scriptsize{1}}}
    \ar@{-}[ld]_{\text{in}}
    \ar@{-}[rrrdd]^(.3){\text{out}}
    \\
    & & \raise0.2ex\hbox{\textcircled{\tiny{CS}}}
    \ar@{-}[ld]_(.7){\text{C}}
    \ar@{-}[rd]^(.7){\text{A}}
    &&&&&
    \\
    & && && &&
    \\
  }}
  \]
\caption{A tree of the constituent game of a chain store game.\label{con_cs}}
\end{figure}
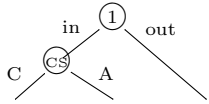

The chain store is set to prefer (out) to (in, C), while the local store is set to  prefer (in, C) to (out).
Both players prefer (in, A) the least.
\[
    (\text{out})\succsim_{\text{CS}}(\text{in},\text{C})\succsim_{\text{CS}}(\text{in},\text{A})\hspace{1cm}
    (\text{in},\text{C})\succsim_{\text{1}}(\text{out})\succsim_{\text{1}}(\text{in},\text{A}).
\]

The subgame perfect equilibrium of the game is uniquely given as $s_1^*(\emptyset)=\text{in}$ and $s_{\text{C}}^*(\text{in})=\text{C}$, since the chain store prefers $(\text{in},\text{C})$ to $(\text{in},\text{A})$, and the local store prefers $(\text{in},\text{C})$ to $(\text{out})$. 

Now let us extend this game to the 2-repeated game with weakly separable preferences.
The tree of the game is shown in Figure \ref{2_cs}.
\begin{figure}[htbp]
  \[
  \centerline{
    \xymatrix@=8.5pt@M=-.10pt{
      & && && &&& &&
      \raise0.2ex\hbox{\textcircled{\scriptsize{1}}}
      \ar@{-}[llld]_{\text{in}}
      \ar@{-}[rrrrrrrrdd]^(.28){\text{out}}
      \\
      & && && &&
      \raise0.2ex\hbox{\textcircled{\tiny{CS}}}
      \ar@{-}[lllld]_{\text{C}}
      \ar@{-}[rrrd]^{\text{A}}
      \\
      & &&
      \raise0.2ex\hbox{\textcircled{\scriptsize{2}}}
      \ar@{-}[ld]_{\text{in}}
      \ar@{-}[rrrdd]^(.28){\text{out}}
      && && && &
      \raise0.2ex\hbox{\textcircled{\scriptsize{2}}}
      \ar@{-}[ld]_{\text{in}}
      \ar@{-}[rrrdd]^(.28){\text{out}}
      && && && &&
      \raise0.2ex\hbox{\textcircled{\scriptsize{2}}}
      \ar@{-}[ld]_{\text{in}}
      \ar@{-}[rrrdd]^(.28){\text{out}}
      \\
      & &
      \raise0.2ex\hbox{\textcircled{\tiny{CS}}} \ar@{-}[ld]_(.7){\text{C}} \ar@{-}[rd]^(.7){\text{A}}
      & && && &&
      \raise0.2ex\hbox{\textcircled{\tiny{CS}}} \ar@{-}[ld]_(.7){\text{C}} \ar@{-}[rd]^(.7){\text{A}}
      && && && &&
      \raise0.2ex\hbox{\textcircled{\tiny{CS}}} \ar@{-}[ld]_(.7){\text{C}} \ar@{-}[rd]^(.7){\text{A}}
      \\
      &  && && && && && && && && && &&
      \\\\
      \save[]+<4.5cm,0cm>*\txt<18pc>{\scriptsize A tree of the 2-repeated chain store game}\restore
  }}
  \]

  \caption{A tree of the 2-repeated chain store game.\label{2_cs}}
\end{figure}

As can be seen from Figure \ref{2_cs}, the set of all connected terminal histories is given by $C=\{(\textrm{out}),(\textrm{in},\textrm{C}),(\textrm{in},\textrm{A})\}$, which coincides with the set $Z$ of all terminals.
The set of all terminal whole histories is given by
\[
\mathbf{Z}^2 \ = \
\left\{
\begin{matrix}
  ((\mathrm{in, C}),(\mathrm{in, C})),\
  ((\mathrm{in, C}),(\mathrm{in, A})),\
  ((\mathrm{in, C}),(\mathrm{out})),\ \\
  ((\mathrm{in, A}),(\mathrm{in, C})),\
  ((\mathrm{in, A}),(\mathrm{in, A})),\
  ((\mathrm{in, A}),(\mathrm{out})),\ \\
  \hspace{.7cm}((\mathrm{out}),(\mathrm{in, C})),\
  ((\mathrm{out}),(\mathrm{in, A})),\
  ((\mathrm{out}),(\mathrm{out}))
\end{matrix}
\right\}.
\]

The preference relations induced by the constituent game are represented by the Hasse diagrams below:

\begin{figure}[htbp]
  \[
  \centerline{
    \xymatrix@R=14pt@C=-22pt@M=2pt{
      && \text{\tiny$((\text{out}),(\text{out}))$}
      \ar@{-}[ld]
      \ar@{-}[rd]
      \\
      & \text{\tiny$((\text{out}),(\text{in},\text{C}))$}
      \ar@{-}[ld]
      \ar@{-}[rd]
      && \text{\tiny$((\text{in},\text{C}),(\text{out}))$}
      \ar@{-}[ld]
      \ar@{-}[rd]
      \\
      \text{\tiny$((\text{out}),(\text{in},\text{A}))$}
      \ar@{-}[rd]
      && \text{\tiny$((\text{in},\text{C}),(\text{in},\text{C}))$}
      \ar@{-}[ld]
      \ar@{-}[rd]
      && \text{\tiny$((\text{in},\text{A}),(\text{out}))$}
      \ar@{-}[ld]
      \\
      & \text{\tiny$((\text{in},\text{C}),(\text{in},\text{A}))$}
      \ar@{-}[rd]
      && \text{\tiny$((\text{in},\text{A}),(\text{in},\text{C}))$}
      \ar@{-}[ld]
      \\
      && \text{\tiny$((\text{in},\text{A}),(\text{in},\text{A}))$}
      \\
      \save[]+<2cm,0cm>*\txt<14pc>{\scriptsize (a) Preference relations of the chain store}\restore
    }
    \hspace{.4cm}
    \xymatrix@R=14pt@C=-22pt@M=2pt{
      && \text{\tiny$((\text{in},\text{C}),(\text{in},\text{C}))$}
      \ar@{-}[ld]
      \ar@{-}[rd]
      \\
      & \text{\tiny$((\text{in},\text{C}),(\text{out}))$}
      \ar@{-}[ld]
      \ar@{-}[rd]
      && \text{\tiny$((\text{out}),(\text{in},\text{C}))$}
      \ar@{-}[ld]
      \ar@{-}[rd]
      \\
      \text{\tiny$((\text{in},\text{C}),(\text{in},\text{A}))$}
      \ar@{-}[rd]
      && \text{\tiny$((\text{out}),(\text{out}))$}
      \ar@{-}[ld]
      \ar@{-}[rd]
      && \text{\tiny$((\text{in},\text{A}),(\text{in},\text{C}))$}
      \ar@{-}[ld]
      \\
      & \text{\tiny$((\text{out}),(\text{in},\text{A}))$}
      \ar@{-}[rd]
      && \text{\tiny$((\text{in},\text{A}),(\text{out}))$}
      \ar@{-}[ld]
      \\
      && \text{\tiny$((\text{in},\text{A}),(\text{in},\text{A}))$}
      \\
      \save[]+<2cm,0cm>*\txt<12pc>{\scriptsize (b) Preference relations of the local store}\restore
  }  }  
\]

  \caption{Hasse diagrams of the preference relations of (a) the chain store and (b) the local store.\label{2_cs_pref}}
\end{figure}

Each line segment represents a preference relation between its two ends, i.e. the upper node is preferred to the lower node.
For example, the line segment between $((\text{out}),(\text{out}))$ and $((\text{in},\text{C}),(\text{out}))$ represents a relation $((\text{out}),(\text{out})) \succsim_{\text{CS}}^2 ((\text{in},\text{C}),(\text{out}))$ which is the result of weak separability with $(\text{out}) \succsim_{\text{CS}} (\text{in},\text{C})$.


It is important to remember that the diagrams do not form a chain.
The relationship between e.g. $((\text{out}),(\text{in},\text{A}))$ and $((\text{in},\text{A}),(\text{out}))$ is not determined.
This indicates that there are not enough constraints to determine whether one is preferred to the other terminal node or not. 
This problem of missing line segments will remain until payoff functions are specified.

Nevertheless, these line segments are sufficient to determine the subgame perfect equilibrium.
Since all players are assumed to know all previously chosen actions, the subgame perfect equilibrium of the game is uniquely given as $\mathbf{s}^*_t(\mathbf{h})=\textrm{in}$ if $\mathbf{P}^2(\mathbf{h})=t$ for both local stores, and $\mathbf{s}_{\text{CS}}^*(\mathbf{h})=\textrm{C}$ if $\mathbf{P}^2(\mathbf{h})=\textrm{CS}$, regardless of the whole history $\mathbf{h}$. 
\end{ex}

As shown in Example \ref{cs}, repeating the equilibrium strategy $s^*$ of the constituent game $\Gamma$ twice in the 2-repeated chain store game was confirmed to be the unique subgame perfect equilibrium of the overall game.
It may seem that this result is always guaranteed.
However, this is not always the case, particularly when the class $C$ of connected terminal histories is strictly smaller than the class $Z$ of all terminal histories.
Consider the centipede games, for example. 
Suppose an equal amount of bonuses is added for both players equally among the terminal nodes in the second stage.
This prompts them to change their behavior and continue to the second stage instead of quitting at the first stage. 
In this case, the unique subgame perfect strategy is no longer preserved.

To avoid this problem, additional conditions must be met.
These conditions are implicitly satisfied in normal contipede games.
The next proposition specifies this.

\begin{prop}\label{ext}
  Suppose $s^*$ is a subgame perfect equilibrium of $\Gamma$ and $\succsim_i^\tau$ is weakly separable and hugely transitive.
  Then ${s}^{*\tau}$ is also a subgame perfect equilibrium of $\mathbf{\Gamma}^\tau$ if the indifference $(h_c)^\frown \left(O\left({{s}^{*}}\right)\right) \sim_i^\tau (h_c)$ for all $i\in I$ and all $h_c\in C$, called \textit{dynamic consistency}, is satisfied. 
\end{prop}

\begin{proof}
  Let $\mathbf{h}=(h_1,\ldots,h_t)$ be an arbitrary non-terminal or connected terminal whole history, i.e. $\mathbf{h}\notin \mathbf{Z}^\tau\setminus \mathbf{C}^\tau$, and $i=\mathbf{P}^\tau(\mathbf{h})$.
  Let $j_{h}^*$ denote the equilibrium outcome $O_{{h}}(s^*|_{{h}})$ of the subgame $\Gamma({h})$.
  Let $\mathbf{s}_{i}$ be an arbitrarily chosen strategy and $(j_{h_t},j_{|\mathbf{h}|+1},\ldots,j_{\tau})=\mathbf{O}_{\mathbf{h}}(\mathbf{s}_{i}|_{\mathbf{h}},{s}^{*\tau}_{-i}|_{\mathbf{h}})$.
  If a whole history $(h_1,\ldots,h_{t-1})^\frown(h_t^\frown j_{h_t})^\frown(j_{|\mathbf{h}|+1},\ldots,j_{\lambda})$ is terminal for some $\lambda\geq|\mathbf{h}|+1$, then the history $j_\ell$ is set to be empty for all $\ell\in\{\lambda+1,\ldots,\tau\}$. 
  
  Suppose $C=Z$.
  Then, $j^*_{{h_t}}\, {\succsim_{i}}|_{{h_t}}\, j_{{h_t}}$ holds for every strategy $\mathbf{s}_{i}$, since $s^*$ is a subgame perfect equilibrium of $\Gamma$.
  This implies that, by weak separability, the following relation holds
  \[
  (j_{{h_t}}^*)^{\frown} (j_{|\mathbf{h}|+1},\ldots,j_\tau)
  \  {\succsim_{i}^{\tau}}|_{\mathbf{h}}\ 
  (j_{{h_t}})^{\frown} (j_{|\mathbf{h}|+1},\ldots,j_\tau).
  \]
  The following also holds by weak separability for any $\ell\in \{0,\ldots,\tau-|\mathbf{h}|-1\}$:
  \[
  (j_{{h_t}}^*)^{\frown} j_\emptyset^{*\ell+1\frown} (j_{|\mathbf{h}|+\ell+2},\ldots,j_\tau)
  \  {\succsim_{i}^{\tau}}|_{\mathbf{h}}\ 
  (j_{{h_t}}^*)^{\frown} j_\emptyset^{*\ell\frown} (j_{|\mathbf{h}|+\ell+1},\ldots,j_\tau).
  \]
  By huge transitivity, $\mathbf{O}_{\mathbf{h}}(s^{*\tau}|_{\mathbf{h}})\, {\succsim_{i}^{\tau}}|_{\mathbf{h}}\, \mathbf{O}_{\mathbf{h}}(\mathbf{s}_{i}|_{\mathbf{h}},{s}^{*\tau}_{-i}|_{\mathbf{h}})$ holds.

  Suppose, on the other hand, that $Z\ne C$.
  The case where  $h_t^\frown j_{h_t}^*\in {C}$ is satisfied can be proved in exactly the same way as in the case $C=Z$.
  Concerning the opposite case $h_t^\frown j_{h_t}^*\notin {C}$, the proof can be divided into two parts, depending on whether $O(s^*)$ is a connected terminal or not.
  Let us start with the case where it is, i.e. $O(s^*)=j^*_\emptyset\in C$.
  Then, the relation $(j^*_{h_t}) \succsim^{\tau}_{i}|_{\mathbf{h}} (j_{h_t})^\frown j_{\emptyset}^{*\tau-|\mathbf{h}|}$ holds if $h_t^\frown j_{h_t}\in C$, since the indifference $(j_{h_t})^\frown j_\emptyset^{*\tau-|\mathbf{h}|} {{\sim_{i}^{\tau}}|_{\mathbf{h}}} (j_{h_t})$ is satisfied by dynamic consistency $(h_c)^\frown{O(s^*)}\sim_i^\tau(h_c)$ with the weak separability and huge transitivity.
  Since $j_\emptyset^*\ {{\succsim_{i}}}\ j_\ell$ holds for all $\ell\in\{|\mathbf{h}|+1,\ldots,\tau\}$, the relation $(j_{h_t})^\frown j_\emptyset^{*\tau-|\mathbf{h}|}\, {{\succsim_{i}^{\tau}}|_{\mathbf{h}}}\, (j_{h_t})^\frown (j_{|\mathbf{h}|+1},...,j_\tau)$ also holds by weak separability, and the following relation holds
  \[
  \mathbf{O}_{\mathbf{h}}(s^{*\tau}|_{\mathbf{h}})
  \, =\,
  (j_{h_t}^*)
  \ {{\succsim_{i}^\tau}|_{\mathbf{h}}}\
  (j_{h_t})^\frown  (j_{|\mathbf{h}|+1},\ldots,j_\tau)
  \, =\,
  \mathbf{O}_{\mathbf{h}}(\mathbf{s}_{i}|_{\mathbf{h}},s^{*\tau}_{-i}|_{\mathbf{h}}).
  \]

  Second, let us confirm the case where $O(s^*)= j^*_\emptyset\in Z\setminus C$ is satisfied.
  Since $j^*_\emptyset\succsim_{\emptyset} j_\ell$ holds for all $\ell\in\{|\mathbf{h}|+1,...,\tau\}$, the following relation also holds for all $\ell\in\{2,\ldots,\tau-|\mathbf{h}|\}$ by weak separability
  \[
  (j_{h_t})^\frown(j_{|\mathbf{h}|+1},\ldots,j_{|\mathbf{h}|+\ell-1}){}^\frown (j^*_\emptyset)
  \succsim_{i}^\tau\hspace{-1mm}|_{\mathbf{h}}\ 
  (j_{h_t})^\frown(j_{|\mathbf{h}|+1},\ldots,j_{|\mathbf{h}|+\ell}).
  \]
  Since $(j_{h_t})^\frown(j_{|\mathbf{h}|+1},\ldots,j_{|\mathbf{h}|+\ell-1})^\frown (j^*_\emptyset)\sim_{i}^\tau\hspace{-1mm}|_{\mathbf{h}}\, (j_{h_t})^\frown(j_{|\mathbf{h}|+1},\ldots,j_{|\mathbf{h}|+\ell-1}){}$ holds for all $\ell\in\{1,\ldots,\tau-|\mathbf{h}|-1\}$ by dynamic consistency and weak separability, the following relation also holds for all $\ell\in\{2,\ldots,\tau-|\mathbf{h}|-1\}$
  \[
  (j_{h_t})^\frown(j_{|\mathbf{h}|+1},\ldots,j_{|\mathbf{h}|+\ell-1}){}^\frown (j^*_\emptyset)
  \succsim_{i}^\tau\hspace{-1mm}|_{\mathbf{h}}\ 
  (j_{h_t})^\frown(j_{|\mathbf{h}|+1},\ldots,j_{|\mathbf{h}|+\ell}){}^\frown (j^*_\emptyset).
  \]
  Since $(j_{h_t})^\frown(j^*_\emptyset)\sim_{i}^\tau\hspace{-1mm}|_{h_t}(j_{h_t})$ and $(j^*_{h_t})\succsim_{i}^\tau\hspace{-1mm}|_{\mathbf{h}} (j_{h_t})$ by assumption, 
  the following holds by huge transitivity
  \[
  \mathbf{O}_{\mathbf{h}}(s^{*\tau}|_{\mathbf{h}})
  \, =\,
  (j_{h_t}^*)
  \ {{\succsim_{i}^\tau}|_{\mathbf{h}}}\
  (j_{h_t})^\frown  (j_{|\mathbf{h}|+1},\ldots,j_\tau)
  \, =\,
  \mathbf{O}_{\mathbf{h}}(\mathbf{s}_{i}|_{\mathbf{h}},s^{*\tau}_{-i}|_{\mathbf{h}}).
  \]
  
  Now it is confirmed that $s^{*\tau}$ is a subgame perfect equilibrium.
\end{proof}


Depending on the structure of $\Gamma$, the dynamic consistency condition given in Proposition \ref{ext} can be relaxed.
For example, if $\Gamma$ has a subgame perfect equilibrium $s^*$ leading to $O(s^*)\in C$, 
then the condition $(h_c)^\frown \left(O({s}^*)\right)\sim^\tau_i (h_c)$ can be relaxed to $(h_c)^\frown \left(O({s}^*)\right)\succsim_i^\tau (h_c)$.
Conversely, if $O(s^*)\notin C$, it can be replaced by $(h_c)^\frown \left(O({s}^*)\right)\precsim_i^\tau (h_c)$. 
The next example illustrates the circumstances in which this condition is required or can be relaxed.

\begin{ex}[Centipede Games]\label{dc}
The centipede game was proposed by Rosenthal \cite{centipede} to highlight the problem inherent in the concept of subgame perfect equilibria more clearly than the chain store paradox \cite{chainstore}.

The game is played between two players.
Player 1 first chooses between R to continue or D to quit the game.
After 1's move, 2 decides whether to take r or d if the game does not end.
The game continues until one of the players chooses D or d to terminates the game or they reach the last node.

This game can also be built from the 2-player constituent game.
Unlike the chain store games, both players are core in this case.
The set of connected histories is given as $C=\{\mathrm{Rr}\}$ (where Rr stands for $(\text{R},\text{r})$), which is a proper subset of $Z=\{\mathrm{D}, \mathrm{Rd},\mathrm{Rr}\}$.
Figure \ref{con_cp} shows an example of the game's trees.
%

\begin{figure}[htbp]
  \[
  \centerline{
    \xymatrix@=25pt@M=-.10pt{
      \raise0.2ex\hbox{\textcircled{\scriptsize{1}}}\ar@{-}[r]^{\mathrm{R}}\ar@{-}[d]^{\mathrm{D}}
      &
      \raise0.2ex\hbox{\textcircled{\scriptsize{2}}}\ar@{-}[r]^{\mathrm{r}}\ar@{-}[d]^{\mathrm{d}}
      &
      \\ 
      & & 
  }}
  \]
  \caption{A tree of the constituent game of a centipede game.\label{con_cp}}
\end{figure}
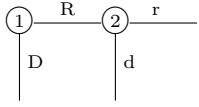

Player 1 prefers Rr to D and D to Rd. Player 2, on the other hand, prefers Rd to Rr and Rr to D. 
\[
\text{Rr}\succsim_1\text{D}\succsim_1\text{Rd}\hspace{1cm}
\text{Rd}\succsim_2\text{Rr}\succsim_2\text{D}.
\]
Then the profile of strategies that constitutes the unique subgame perfect equilibrium is given by $s^*_1(\emptyset)=\mathrm{D}$ and $s^*_2(\mathrm{R})=\mathrm{d}$.

Extending the game to multiple periods with weakly separable preferences, one has to face the problem that arises when $C$ is strictly smaller than $Z$.
In this situation, one has to determine preference relations between the terminal nodes of different lengths.
This is where the dynamic consistency condition comes in.
The tree of the game repeated twice is shown in Figure \ref{2_cp}.

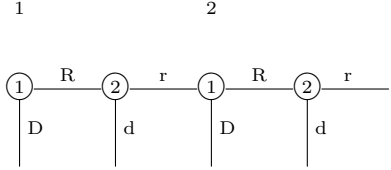
\begin{figure}[htbp]
  \[
  \centerline{
    \xymatrix@=25pt@M=-.10pt{
      \text{\scriptsize{1}}
      &&
      \text{\scriptsize{2}}
      &&
      \\
      \raise0.2ex\hbox{\textcircled{\scriptsize{1}}}\ar@{-}[r]^{\mathrm{R}}\ar@{-}[d]^{\mathrm{D}}
      & 
      \raise0.2ex\hbox{\textcircled{\scriptsize{2}}}\ar@{-}[r]^{\mathrm{r}}\ar@{-}[d]^{\mathrm{d}}
      & 
      \raise0.2ex\hbox{\textcircled{\scriptsize{1}}}\ar@{-}[r]^{\mathrm{R}}\ar@{-}[d]^{\mathrm{D}}
      & 
      \raise0.2ex\hbox{\textcircled{\scriptsize{2}}}\ar@{-}[r]^{\mathrm{r}}\ar@{-}[d]^{\mathrm{d}}
      &
      \\
      & & & & 
  }}
  \]
  \caption{A tree of the 2-repeated centipede game.\label{2_cp}}
\end{figure}

As already confirmed, the path of the unique subgame perfect equilibrium of the constituent game is given by $O(s^*)=\text{D}$.
This implies that dynamic consistency requires preference relations to satisfy $(\text{Rr})\sim_i^2 (\text{Rr},\text{D})$ for all $i\in I$.
Taking into account of the restriction, preference relations at all terminal nodes of 2-repeated games $\mathbf{\Gamma}^2$ can be represented as the Hasse diagrams shown below.

\begin{figure}[htbp]
  \[
  \centerline{
    \xymatrix@R=16pt@C=-15pt@M=2pt{
      &&
      \text{\scriptsize{(Rr, Rr)}}
      \ar@{-}[ld]
      \\
      &
      \text{\scriptsize{(Rr), (Rr, D)}}
      \ar@{-}[ld]
      \ar@{-}[rd]
      \\
      \text{\hspace{.4cm}\scriptsize{(D)}\hspace{.4cm}}
      \ar@{-}[rd]
      &&
      \text{\hspace{.4cm}\scriptsize{(Rr, Rd)}}
      \ar@{.>}@<1ex>[lu]^{\text{\tiny 2nd BI}}
      \\
      &
      \text{\scriptsize{(Rd)}}
      \ar@{.>}@<1ex>[lu]^{\text{\tiny last BI}}
      \\
      \save[]+<.8cm,0cm>*\txt<14pc>{\scriptsize (a) Preference relations of Player 1}\restore
    }
    \hspace{0cm}
    \xymatrix@R=16pt@C=-15pt@M=2pt{
      &
      \text{\scriptsize{(Rr, Rd)}\hspace{.4cm}}
      \ar@{-}[rd]
      \\
      \text{\hspace{.2cm}\scriptsize{(Rd)}\hspace{.2cm}}      
      \ar@{-}[rd]
      &&
      \text{\scriptsize{(Rr, Rr)}}
      \ar@{-}[ld]
      \ar@{.>}@<-1ex>[lu]_{\text{\tiny 1st BI}}
      \\
      &
      \text{\scriptsize{(Rr), (Rr,D)}}
      \ar@{-}[ld]
      \ar@{.>}@<-1ex>[lu]_{\text{\tiny 3rd BI}}
      \\
      \text{\hspace{.3cm}\scriptsize{(D)}\hspace{.3cm}}
      \\
      \save[]+<.8cm,0cm>*\txt<14pc>{\scriptsize (b) Preference relations of Player 2}\restore
  }  }
  \]
  \caption{Hasse diagrams of the preference relations of (a) the player 1 and (b) the player 2 (``BI'' is an abbreviation for backward induction).\label{2_cp_pref}}
\end{figure}
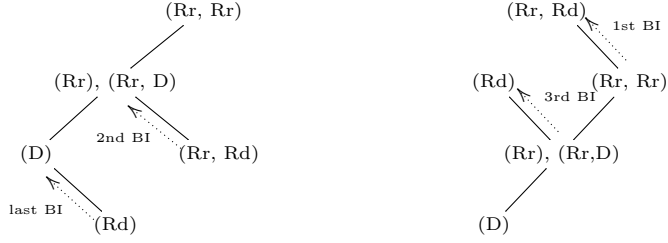

As can be seen from these diagrams, the unique subgame perfect equilibrium is given by ${s}^{*2}_1(\mathbf{h})=\mathrm{D}$ for all $\mathbf{P}^2(\mathbf{h})=1$ and ${s}^{*2}_2(\mathbf{h})=\mathrm{d}$ for all $\mathbf{P}^2(\mathbf{h})=2$.
It is also clear that the result remains valid even if some of the the conditions are relaxed, as long as the condition $(h)^\frown(O(s^*))\precsim^2_i(h)$ is satisfied. 
In this example, this corresponds to the case where $(\text{Rr, D})\precsim_i^2\text{(Rr)}$ is satisfied.

If dynamic consistency is not satisfied, $s^{*2}$ is no longer guaranteed to be a subgame perfect equilibrium.
For example, if the condition $(\text{Rr},\text{D})\succ_2^2(\text{Rr})$ is imposed, instead,
it may happen that $(\text{Rr},\text{D})\succ_2^2(\text{Rd})$ holds, so that the equilibrium strategy of the player 2 
changes to $\mathbf{s}^*_2(\mathbf{\text{(R)}})=\text{r}$, 
and $s^{*2}$ is no longer a subgame perfect equilibrium.
To maintain the equilibrium, the condition $\left(\text{Rr},\text{D}\right)\ \precsim_2^2\ \left(\text{Rr}\right)$ must be met at least.
\end{ex}

\section{Strategic Games}

By slightly modifying Definition \ref{def_ext}, a \textit{strategic game} can be described as follows.

\begin{defn}\label{def_st}[Definition 11.1 of \cite{OR}]
  A \textit{strategic game } has the following components
  \begin{itemize}
  \item A class $I$ of players.
  \item A class $H$ of histories which consists of an initial history $\emptyset$ and $|I|$-tuples of actions denoted by $(a_i)_{i\in I}\in Z$, where $Z$ is a class of all terminals $\times_{i\in I}A_i$.
    $A_i$ denotes the class of actions available to player $i$.
  \item For each player $i\in I$ a preference relation $\succsim_i$ on $H$.
  \end{itemize}
\end{defn}

The triple $G=\langle I, H, (\succsim_i)\rangle$ is called a \textit{strategic game}.

Note that in extensive games with perfect information, each history is a sequence of individual actions taken one at a time in each period, while in strategic games it is an $|I|$-tuple of actions taken all at once.
This is why classes of actions are distinguished by players in strategic games.

A strategy $s_i$ of player $i\in I$ in a strategic game can also be defined as a function that assigns each action in $A_i$ to the initial history.
Since its domain consists only of the initial history, it can be simplified and written as $s=(a_i)_{i\in I}$.
For each strategy profile $s=(a_i)_{i\in I}$ in the strategic game, the outcome $O(s)$ of $s=(a_i)_{i\in I}$ is simply $s$ itself.
And Nash equilibrium is given as follows.

\begin{defn}[Definition 14.1 of \cite{OR}]\label{Nash_st}
  A \textit{Nash equilibrium of a strategic game} $\langle I, H, (\succsim_i)\rangle$ is a strategy profile $s^*$ with the property for each player $i\in I$
  \[
  (a^*_{-i}, a^*_i) \succsim_i   (a^*_{-i}, a_i) \text{ for all } a_i \in A_i.
  \]
\end{defn}

Repeated games of strategic games can be stated essentially in the same way as Definition 137.1 of \cite{OR}.

\begin{defn}[Definition 137.1 of \cite{OR}]
  Let $G=\langle I,H,(\succsim_i)\rangle$ be a strategic game.
  Then, a \textit{$\tau$-repeated game of $G$} is an extensive game with perfect information $\mathbf{G}^\tau=\langle {I}, \mathbf{H}^\tau, (\succsim^\tau_i)\rangle$ where
  \begin{itemize}
  \item $\mathbf{H}^\tau=\bigcup_{t\in \{1,\ldots,\tau\}} H^t$ 
\item $\succsim_i^\tau$ is a preference relation on $Z^\tau$. 
  \end{itemize}
\end{defn}

The strategy of player $i$ in a $\tau$-repeated game $\mathbf{G}^\tau$ of ${G}$ is given as follows.

\begin{defn}
  A \textit{strategy $\mathbf{s}_i$ of a player} $i\in I$ in a $\tau$-repeated game $\mathbf{G}^\tau$ of $G$ with perfect information is a function that assigns an action in ${A}_i$ to each nonterminal whole history $\mathbf{h}=(h_1,\ldots, h_k)\in\mathbf{H}^\tau\setminus{Z}^\tau$.
\end{defn}


The outcome $\mathbf{O}(\mathbf{s})$ of $\mathbf{G}^\tau$ is simply given as the terminal whole history $(h_1,\ldots,h_{\tau})\in{Z}^\tau$ that satisfies $(\mathbf{s}_{1}(h_1,\ldots,h_t),\ldots,\mathbf{s}_{|I|}(h_{1},\ldots,h_t))=h_{t+1}$ for all $t\in\{1,\ldots,\tau-1\}$. 
Subgames are also simplified as below.

\begin{defn}
  The \textit{subgame of the $\tau$-repeated game with perfect information} $\mathbf{G}^\tau=\langle {I},\mathbf{H}^\tau,(\succsim_i^\tau)\rangle$ \textit{that follows the whole history $\mathbf{h}$} is the $\tau-|\mathbf{h}|$-repeated game $\mathbf{G}^\tau\mathbf{(h)}=\langle {I},\mathbf{H}^\tau|_{\mathbf{h}},({\succsim_i^\tau}|_\mathbf{h})\rangle$, where $\mathbf{H}^\tau|_\mathbf{h}$ is the set of whole histories $\mathbf{h}'$ that satisfy $\mathbf{h}^\frown \mathbf{h}'\in \mathbf{H}^\tau$, and ${\succsim_i^\tau}|_\mathbf{h}$ is defined by $\mathbf{h}'{\succsim_i^\tau}|_\mathbf{h} \mathbf{h}''$ if and only if $\mathbf{h}^\frown \mathbf{h}'\succsim_i^\tau \mathbf{h}^\frown \mathbf{h}''$.
\end{defn}

The subgame perfect equilibrium is also defined as in the case of $\mathbf{\Gamma}^\tau$.
%

\begin{ex}[The Prisoner's Dilemma]\label{PD}
  The game is played by two players.
  Both players are suspected of having committed a crime, which they did committed.
  If they both remain silent, they will receive a lighter sentence.
  If only one of them confesses, the one can exempt from the sentence and the accomplice who remains silent gets the heaviest sentence.
  If both of them confess, their sentences are relatively light.
  The resulting preference relations are given as below.
  \[
  \text{CS} \succsim_1 \text{SS}  \succsim_1 \text{CC}\succsim_1 \text{SC}\hspace{1cm}
  \text{SC} \succsim_2 \text{SS} \succsim_2 \text{CC} \succsim_2 \text{CS}
  \]
  where C and S abbreviate ``confess'' and ``silence'' respectively.
  

  The Nash equilibrium of the game is $s_1^*=s_2^*=\text{C}$ (short for Confess), and the outcome CC is realised.
  This outcome is disappointing for both players, as they could have avoided it by remaining silent.
  However, the better outcome, in which both players remain silent, is cleverly prevented from occurring, because neither player would benefit from doing so alone.
  This is why the problem is called a ``dilemma''.

  The nature of the problem never changes when the game is repeated, unless it is finite.
  In fact, the subgame perfect equilibrium and its outcome remain the same no matter how many times the game is repeated, as long as it is played a finite number of times.
  It is also clear from the Hasse diagrams of the preference relations of both players induced from the constituent game, shown in Figure \ref{2_pd_pref}, that confession dominates silence, whatever the action of the accomplice.
  
  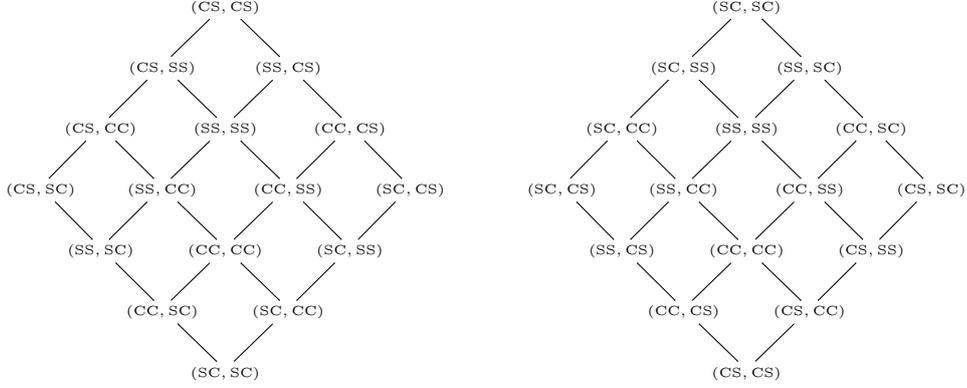
\begin{figure}[htbp]
    \[
    \centerline{
      \xymatrix@R=14pt@C=-8pt@M=2pt{
        &&&
        \text{\tiny$(\text{CS},\text{CS})$}
        \ar@{-}[ld]
        \ar@{-}[rd]
        \\
        &&
        \text{\tiny$(\text{CS},\text{SS})$}
        \ar@{-}[ld]
        \ar@{-}[rd]
        &&
        \text{\tiny$(\text{SS},\text{CS})$}
        \ar@{-}[ld]
        \ar@{-}[rd]
        \\
        &
        \text{\tiny$(\text{CS},\text{CC})$}
        \ar@{-}[ld]
        \ar@{-}[rd]
        &&
        \text{\tiny$(\text{SS},\text{SS})$}
        \ar@{-}[ld]
        \ar@{-}[rd]
        &&
        \text{\tiny$(\text{CC},\text{CS})$}
        \ar@{-}[rd]
        \ar@{-}[ld]
        \\
        \text{\tiny$(\text{CS},\text{SC})$}
        \ar@{-}[rd]
        &&
        \text{\tiny$(\text{SS},\text{CC})$}
        \ar@{-}[rd]
        \ar@{-}[ld]
        &&
        \text{\tiny$(\text{CC},\text{SS})$}
        \ar@{-}[rd]
        \ar@{-}[ld]
        &&
        \text{\tiny$(\text{SC},\text{CS})$}
        \ar@{-}[ld]
        \\
        &
        \text{\tiny$(\text{SS},\text{SC})$}
        \ar@{-}[rd]
        &&
        \text{\tiny$(\text{CC},\text{CC})$}
        \ar@{-}[ld]
        \ar@{-}[rd]
        &&
        \text{\tiny$(\text{SC},\text{SS})$}
        \ar@{-}[ld]
        \\
        &&
        \text{\tiny$(\text{CC},\text{SC})$}
        \ar@{-}[rd]
        &&
        \text{\tiny$(\text{SC},\text{CC})$}
        \ar@{-}[ld]
        \\
        &&&
        \text{\tiny$(\text{SC},\text{SC})$}
        \\
        \save[]+<2cm,0cm>*\txt<14pc>{\scriptsize (a) Preference Relations of player 1}\restore
      }
      \hspace{.4cm}
      \xymatrix@R=14pt@C=-8pt@M=2pt{
        &&&
        \text{\tiny$(\text{SC},\text{SC})$}
        \ar@{-}[ld]
        \ar@{-}[rd]
        \\
        &&
        \text{\tiny$(\text{SC},\text{SS})$}
        \ar@{-}[ld]
        \ar@{-}[rd]
        &&
        \text{\tiny$(\text{SS},\text{SC})$}
        \ar@{-}[ld]
        \ar@{-}[rd]
        \\
        &
        \text{\tiny$(\text{SC},\text{CC})$}
        \ar@{-}[ld]
        \ar@{-}[rd]
        &&
        \text{\tiny$(\text{SS},\text{SS})$}
        \ar@{-}[ld]
        \ar@{-}[rd]
        &&
        \text{\tiny$(\text{CC},\text{SC})$}
        \ar@{-}[rd]
        \ar@{-}[ld]
        \\
        \text{\tiny$(\text{SC},\text{CS})$}
        \ar@{-}[rd]
        &&
        \text{\tiny$(\text{SS},\text{CC})$}
        \ar@{-}[rd]
        \ar@{-}[ld]
        &&
        \text{\tiny$(\text{CC},\text{SS})$}
        \ar@{-}[rd]
        \ar@{-}[ld]
        &&
        \text{\tiny$(\text{CS},\text{SC})$}
        \ar@{-}[ld]
        \\
        &
        \text{\tiny$(\text{SS},\text{CS})$}
        \ar@{-}[rd]
        &&
        \text{\tiny$(\text{CC},\text{CC})$}
        \ar@{-}[ld]
        \ar@{-}[rd]
        &&
        \text{\tiny$(\text{CS},\text{SS})$}
        \ar@{-}[ld]
        \\
        &&
        \text{\tiny$(\text{CC},\text{CS})$}
        \ar@{-}[rd]
        &&
        \text{\tiny$(\text{CS},\text{CC})$}
        \ar@{-}[ld]
        \\
        &&&
        \text{\tiny$(\text{CS},\text{CS})$}
        \\
        \save[]+<2cm,0cm>*\txt<14pc>{\scriptsize (a) Preference Relations of player 2}\restore
    }  }  
    \]
    
    \caption{Hasse diagrams of preference relations of both players.\label{2_pd_pref}}
  \end{figure}

  However, the situation changes drastically when the number of repetitions is increased to a huge number depending on the payoffs yielded from the terminal histories.
  These changes are discussed in detail in Section \ref{Payoff}.
\end{ex}

\begin{ex}[Ultra Long-Term Investment]\label{ulti}
  Consider the situation where an investor has to decide on the following simple investment project.
  It costs a small amount of money in each period, but the return exceeds the total investment if the investment is made over a hugely long period of time.
  The problem is to decide whether to invest in each period.

  This situation can be modelled as a single person's decision problem.
  The constituent game can be written as $I=\{1\}$, $A=\{\text{I},\text{N}\}$, where I and N stand for ``invest'' and ``not invest'' respectively.

    Since revenues only exceed the total costs if huge investments are made, it is clear that no investment should be made if the game is played only once.
  In this case, the preference relation on the class of terminal nodes of the constituent game  is given by $\text{N}\succsim\text{I}$.

  The situation remains the same if the game is repeated only a finite number of times.
  However, it happens that investing whole periods is strictly preferable to all other alternatives when the game is played a huge number of times, since the return on investment exceeds all the costs incurred.
  Note, however, that the result violates weak separability, since $(\text{N})^\frown\text{I}^{\tau-1}\not\succsim^\tau(\text{I})^\frown\text{I}^{\tau-1}$ holds.
\end{ex}

\begin{ex}[Lifestyle Disease]\label{addiction}
  Similar to Example \ref{ulti} is the lifestyle disease problem.
  In this type of problem, a decision is made in each period whether to eat foods or drink beverages that provide short-term benefits but will cause some health problems in the distant future.

  This situation can also be modelled as a single-person decision problem.
  The constituent game can be written as $I=\{1\}$, $A=\{\text{E},\text{A}\}$, where E and A stand for ``eat'' and ``avoid'' respectively.

  Since the health problems that outweigh the short-term benefits  only set in after a very long time, it is clear that the problematic foods will be eaten if the game is played only once, so that the preference relation of the constituent game is given by $\text{E}\succsim\text{A}$.

  However, once the symptoms appear after a huge amount of time has passed, the preference relations can be reversed $(\text{E})^\frown\text{A}^{\tau-1}\not\succsim^\tau(\text{E})^{\frown}\text{A}^{\tau-1}$, 
  so that the weak separability is violated again.
  
\end{ex}

\section{Perspectives on Whole Histories}\label{perspective}
Traditionally, history has been treated from a perspective that provides a fine-grained view of all points in time, and utility over time is evaluated by the weighted sum of these events. 
However, this seems inappropriate when history is made up  of a huge number of events, since most people do not have a complete picture of these events when they occur in a very short period of time or will take place in the distant future. 
For this reason, this paper takes a very different approach to the treatment of history.

With this in mind, the present paper adopts two contrasting perspectives.
Each of these perspectives is represented by a different topology, which serves as a criterion to distinguish what is considered the same or not.
When trying to see the history from the present, through the perspective view, it is easy to distinguish between points in the relatively near future or the end of the history, but difficult when they are in the huge distance from the present or the end. 
On the other hand, when trying to see it from above, from a bird's eye view, each history would appear to be a continuous sequence of events.

For simplicity, let us assume hereafter that $\tau=2^\varepsilon$ for some $\varepsilon\in N\setminus \FN$.

\subsection{Perspective View} 
The first view is given by the indiscernibility equivalence $\widehat{=}$ which is generated by the sequence $(\widehat{R}_k)_{k\in\FN}$ consisting of
\[
\widehat{R}_k=\left\{\langle a,b\rangle\,;\,
\left(a=b\right)\,\vee\,
\left(
a,b\in[2^k,\tau-2^k+1]
\right)
\right\}.
\]
The indiscernibility equivalence $\widehat{=}$ is given by $\bigcap_{k\in\FN} \widehat{R}_k$.

It distinguishes points in time that are only a finite distance away from the present or the end.
All remaining points are equally reduced to one point as a distant future.
This perspective is represented by the continuum $\widehat{\mathscr{C}}=\langle\{1,\ldots,\tau\},\widehat{=}\rangle$, which is called the \textit{perspective view} of $\{1,\ldots,\tau\}$.

To collect all the points that are distinguishable by $\widehat{=}$, let $t_{i}$ denote a position in $\{1,\ldots \tau\}$ where $i\in\FN$ as
\[
t_{i}\ =\
\begin{cases}
  {\tau}/{2} & \text{ if } i=0\\
  \tau-({i-1})/{2} & \text{ if } i\text{ is odd}\\
       {i}/{2} & \text{ otherwise.}
\end{cases}
\]
Then, the class consisting of these points forms a \textit{choice class} of $\widehat{\mathscr{C}}$.
\[
\widehat{\choice{}{}}=\{t_{i}\,;\, i\in\FN\}.
\]
The points in time that are finitely far from the present, or satisfy $\tau-t\in\FN$, are called the \textit{points of the near future}, while those that are finitely far from the end, or satisfy $\tau-t\in\FN$, are called the \textit{points of the near end}.
The points that are hugely far from both the present and the end, or satisfy $t\in\mon{\widehat{=}}{\tau/2}$, are called the \textit{points of the distant future}.

Each pair of the elements contained in $\widehat{\choice{}{}}$ is discernible.
Furthermore, each element $t\in\{1,\ldots,\tau\}$ has only one element in the choice class $\widehat{\choice{}{}}$ that is indiscernible from $t$ itself.
This implies that the class consisting of all monads of points in $\widehat{\choice{}{}}$ covers $\{1,\ldots,\tau\}$ and thus turns out to be a $\sigma$-partition\footnote{
$(S_i)_{i\in\FN}$ is said to be a \textit{$\sigma$-partition} of a class $S$, essentially if $S$ is partitioned into a countable union of subclasses $\cup_{i\in\FN}S_i$.
In this case, since $\mon{}{t_i}\cap\mon{}{t_j}=\emptyset$ for all $i\not= j$ and $\{1,\ldots,\tau\}=\cup_{i\in\FN}\mon{}{t_i}$ are satisfied, $(\mon{}{t_i})_{i\in\FN}$ is in fact a $\sigma$-partition of $\{1,\ldots,\tau\}$.
This condition is essential to construct appropriate measures on $\{1,\ldots,\tau\}$.
See Sakahara and Sato \cite{cogjump} for a more detailed description, .
} of $\{1,\ldots,\tau\}$.
The class is said to be an \textit{appearance of $\{1,\ldots,\tau\}$ by a perspective view},
and is denoted by $\widehat{\mathscr{T}}=\{\mon{\widehat{=}}{t}\,;\, t\in \widehat{\choice{}{}}\}$.

\begin{defn}\label{consistent_perspective}
  A whole history $\textbf{h}^\tau$ is said to be \textit{consistent with the perspective view} iff $h_t=h_{t'}$ for each pair $t,t'$ satisfying $t\ \widehat{=}\ t'$.
  Let $\widehat{\mathbf{\Gamma}}^\tau$ denote the $\tau$-repeated game whose whole histories are restricted to those consistent with the perspective view, and $\widehat{\mathbf{H}}^\tau$ denote the class of these whole histories.
\end{defn}

When calculating the payoffs, it is necessary to find a way to evaluate how many elements are contained in each monad and to approximate these numbers.
However, these values cannot always be determined by some rational numbers, because these classes may be proper.
For example, the size of $\FN$ cannot be determined by any rational number.
To evaluate the size of these classes in a consistent way, Borel approximating functions\footnote{See Sakahara and Sato \cite{kalina_measure} for details} (BAFs for short) play a key role.

BAFs map each class 
to a sequence that approximates the size of the class. 
Since every $\mon{\widehat{=}}{t_i}$ is an equivalence class of $t_i$ with respect to $\widehat{=}$, it can be generated, using a $\pi$-generating sequence $(\widehat{R}_k)_{k\in\FN}$ of $\widehat{=}$, by a sequence $(\widehat{S}_{i,k})$ which satisfies $\widehat{S}_{i,k}=\widehat{R}_k``\{t_{i}\}$ for all $k\geq i$ and, otherwise, $\widehat{S}_{i,k}=\emptyset$.
The sequence $(\widehat{S}_{i,k})$ is said to be a \textit{Borel generating sequence} of $\mon{\widehat{=}}{t_i}$.  

Since the sequence consisting of the size $|\widehat{S}_{i,k}|$ of each element of $(\widehat{S}_{i,k})_{k\in\FN}$ gives a proper approximation of $\mon{}{t_i}$, the sequence $(|\widehat{S}_{i,k}|)_{k\in\FN}$ is said to be a \textit{Borel approximating sequence}.
Then, the Borel approximating function $F_{\widehat{\mathscr{T}}}$ assigns to each $\mon{\widehat{=}}{t_{i}}$ a Borel approximating sequence $(|\widehat{S}_{i,k}|)_{k\in\FN}$ where
\[
|\widehat{S}_{i,k}|\ =\
\begin{cases}
  \tau-2^k \cdot[k\geq 1]\footnotemark & \text{ if } i=0\\
  [2^k\geq i] \cdot[k\geq 1] & \text{ otherwise}
\end{cases}
\]
\footnotetext{$[k\geq 1]$ is an Iverson bracket, meaning that it equals 1 if $k\geq 1$ and 0 otherwise.}
and the cut of each point is given by
\[
\lim_{k\in\FN}{{F_{\widehat{\mathscr{T}}}\bigl(|{\widehat{S}_{i,k}}|\bigr)}}
\ =\  \begin{cases}
  \tau-\FN    & \text{ if } i=0\\
  1 & \text{ otherwise.}
\end{cases}
\]
Finally, a perspective measure $m_{1,F_{\widehat{\mathscr{T}}}}$ of each ${T}\subseteq\widehat{\mathscr{T}}$ on the continuum $\widehat{\mathscr{C}}$ is given by 
\[
  {m}_{1,F_{\widehat{\mathscr{T}}}}({T})\ =\
  \begin{cases}
    \lim_{k\in\FN}\mon{}{\frac{\tau-2^k}{1}}\ =\ 
    \infty & \hspace{-1.5cm}\hfill\text{ if } \mon{}{\tau/2}\in {T}\\
    \lim_{k\in\FN}\mon{}{\frac{|t_i\in T\,;\, 2^k\geq i|}{1}}\ =\ 
    |\{t\,;\,t\in {T}\}| & \text{ otherwise.}
  \end{cases}
\]
For notational convenience, let $\widehat{m}$ denote the measure.

The measure gives equal weight to individual points in time if they are within finite steps of either the present or the end.
On the other hand, those at a huge distance from both ends degenerate into a single point of infinite weight because they cannot be distinguished.
This is a characteristic of the way history looks when viewed from a  perspective view.

As will become clear later, this huge mass will have different effects on equilibrium behavior, depending on how future payoffs are valued.

\subsection{Bird's Eye View} 
The second view is given by another indiscernibility equivalence $\circeq$  on $\{1,\ldots,\tau\}$ defined as $t\circeq t'$ if and only if $\frac{t}{\tau}\doteq\frac{t'}{\tau}$, and $\mon{\circeq}{t}$ denotes a monad of this equivalence class, whose generating sequence is given by 
\[
\circtop{R}_{k}\ =\ \left\{
\langle a,b\rangle \,;\, 
  {\textstyle\left|\frac{a-b}{\tau}\right|<\frac{1}{2^k}}
\right\}.
\]

Contrary to the first view, in this perspective all the points in time appear to be evenly and continuously distributed over the whole area, as if looking down on the history $\{1,\ldots,\tau\}$ from far above.
Thus, the continuum $\circtop{\mathscr{C}}=\langle\{1,\ldots,\tau\},\circeq\rangle$ representing this perspective is called the \textit{bird's eye view} of $\{1,\ldots,\tau\}$.

To collect all the points that are distinguishable by $\circeq$,
let,  $t_{(i,j)}$ denote a position in $\{1,\ldots,\tau\}$ where $i\in\FN$, $j\in j(i)$ in which $j(0)=j(1)=1$ and  $j(i)=2^{i-2}$ for all $i\geq 2$ as 
\[
t_{(i,j)}\ =\
\begin{cases}
  \tau^i  & \text{ if } i<2\\
  \frac{2j+1}{2^{i-1}}\cdot\tau  & \text{ otherwise.}
\end{cases}
\]
The class $\circtop{\choice{}{}}$ composed of all these points is a choice class of $\circtop{\mathscr{C}}$.
\[
\circtop{\choice{}{}}=\{t_{(i,j)}\,;\, \left((i<2)\wedge(j=0)\right)\vee(\exists i\in\FN\setminus 2)(j\in 2^{i-2})\}.
\]
The class consisting of the monads of these points is called an \textit{appearance of $\{1,\ldots,\tau\}$ by a bird's eye view} denoted by $\circtop{\mathscr{S}}=\{\mon{\circeq}{t}\,;\, t\in\circtop{\choice{}{}}\}$.
Each monad contained in the class $\circtop{\mathscr{S}}$ is, in contrast to the perspective view, of the same size and equally distributed over $\{1,\ldots,\tau\}$.

$\circtop{\mathscr{S}}$ is also a $\sigma$-partition of $\{1,\ldots,\tau\}$ since $\mon{\circtop{=}}{{t}}\cap\mon{\circtop{=}}{{t'}}=\emptyset$ for all ${t}\not\doteq {t'}$ and $\{1,\ldots,\tau\}=\cup_{t\in\circtop{\choice{}{}}} \mon{\circtop{=}}{t}$.

\begin{defn}\label{consistent_bird}
  A whole history $\textbf{h}^\tau$ is said to be \textit{consistent with the bird's eye view} iff $h_t=h_{t'}$ for every pair $t,t'$ satisfying $t \circeq t'$.
  Let $\circtop{\mathbf{\Gamma}^\tau}$ denote the $\tau$-repeated game with the whole histories restricted to those consistent with the perspective view, and $\circtop{\mathbf{H}^\tau}$ denote the class of these whole histories.
\end{defn}

Since each monad $\mon{\circtop{=}}{t}$ is an equivalence class of $\circtop{=}$, the monad of each $t_{(i,j)}\in\circtop{\choice{}{}}$ can be generated by a sequence $(S_{(i,j),k})$ where $S_{(i,j),k}=\circtop{R}_k``\bigl\{{t_{(i,j)}}\bigr\}$ for all $k\geq i$ and, otherwise, $S_{(i,j),k}=\emptyset$.
Let $F_{\circtop{\mathscr{S}}}$ be a BAF on $\circtop{\mathscr{S}}$.
The approximating sequence for each element $\mon{\circtop{=}}{t_{(i,j)}}$ of $\circtop{\mathscr{T}}$ is given by $F_{\circtop{\mathscr{S}}}(\mon{\circtop{=}}{t_{(i,j)}})\ =\ (|S_{(i,j),k}|)_{k\in\FN}$ where
\[
|S_{(i,j),k}|\ =\ \frac{\tau}{2^{k-1+[i<2]}}\cdot [k\geq i]. 
\]

The rational cut\footnote{A rational cut $\frac{A}{B}$ represents a ratio of a cut $A$ to a cut $B$.
Rational numbers can only represent the numbers both of whose numerator and denominator are set numbers.
Rational cuts extend both of their numerators and denominators to be Borel classes.
A typical example of that numbers is $\frac{1}{\FN}$, which cannot be represented by rational numbers.
See details for Sakahara and Sato \cite{kalina_measure}.} of $\mon{\circtop{=}}{t_{(i,j)}}$ measured against $\{1,\ldots,\tau\}$ is given by
\[
\lim_{k\in\FN} {\frac{F_{\circtop{\mathscr{S}}}\bigl({\mon{\circtop{=}}{t_{(i,j)}}}\bigr)}{F_{\circtop{\mathscr{S}}}(\{1,\ldots,\tau\})}}
\ =\  \lim_{k\in\FN}{\frac{\frac{\tau}{2^{k-1+[i<2]}}\cdot [k\geq i]}{{\tau}}} 
\ =\ {\frac{1}{\FN}}.
\]
The rational cut of the interval $\{\ell+1,\ldots,m\}$ where $\frac{m-\ell}{\tau}> c$ for some $c\in\FQ$ is also given by 
\begin{eqnarray*}
\lim_{k\in\FN}\frac{F_{\circtop{\mathscr{S}}}\left({\sum_{i\in\FN}\sum_{\ell<\frac{j}{2^{i-1}}\cdot\tau\leq m }\mon{\circtop{=}}{t_{(i,j)}}}\right)}{F_{\circtop{\mathscr{S}}}(\tau)}\\
&\hspace{-3.2cm} =\ \displaystyle\lim_{k\in\FN}{{\frac{g(k)\cdot \tau/{2^{k-1}}}{\tau}}}
\ =
\ \frac{m-\ell}{\tau},
\end{eqnarray*}
where $g(k)=\left|\left\{ j\in 2^{k-1}\,;\, \ell<\tau\cdot{j}/{2^{k-1}}\leq m \right\}\right|$ is an approximate number of monads contained in the interval.

A bird's eye measure on the continuum $\langle\{1,\ldots,\tau\},\circtop{=}\rangle$ is given by 
\[
m_{\tau,F_{\circtop{T}}}(T)\ =\ \mon{}{\frac{m-\ell}{\tau} }
\]
for each $T=\{\ell+1,\ldots,m\}$.
For ease of notation, let $\circtop{m}$ denote the measure.
Note that, unlike the previous case, this is a probability measure.

\section{Payoff Functions of $\mathbf{\Gamma}^\tau$}\label{Payoff}

Preference relations on terminal whole histories are decided by the way each component history is evaluated.
Subgame perfect equilibria are affected by these evaluations.
It was already shown in Proposition \ref{ext} that $s^{*\tau}$ is a subgame perfect equilibrium of the game $\mathbf{\Gamma}^\tau$ under certain conditions.
However, there are a variety of subgame perfect equilibria other than $s^{*\tau}$ depending on how the huge whole history is perceived by the agents and how their payoff functions are constructed.
To see what kind of new equilibria emerge, 
let us focus our attention on three types of criteria.
Note that the dynamic consistency condition is assumed throughout this section.

\subsection{Discounted/Simple Sum}

The profile of preference relations $(\succsim_i^\tau)$ of generalised repeated games is said to follow the \textit{discounted sum} criterion of constituent games if there exists a finite rational $\delta\in(0,1]\cap\FQ$ (the \textit{discount factor}) such that the sequence $(v_t)_{t\in\{1,\ldots,\tau\}}$ of payoffs is preferred to the sequence $(w_t)_{t\in\{1,\ldots,\tau\}}$ if and only if
\[
\int_{\widehat{\mathscr{T}}}\delta^{t-1}{v_t}\ d\widehat{m}(t) \ \geq \ \int_{\widehat{\mathscr{T}}}\delta^{t-1}{w_t}\ d\widehat{m}(t)
\]
or equivalently 
\[
\begin{matrix}
{\displaystyle\lim_{k\in\FN}\mon{}{\sum_{t\leq 2^k}(\underset{\text{near future}}{\underbrace{\delta^{t-1}v_t}})+\underset{\text{distant future}}{\underbrace{(\tau-2^{k+1})\delta^{\tau/2-1}v_{\tau/2}}}+\sum_{t\leq 2^k}(\underset{\text{near end}}{\underbrace{\delta^{\tau-t}v_{\tau-t+1}}})}}\hfill\\
\hspace{0.5cm}\hfill \geq 
\displaystyle{\lim_{k\in\FN}\mon{}{\sum_{t\leq 2^k}(\delta^{t-1}w_t)+(\tau-2^{k+1})\delta^{\tau/2-1}w_{\tau/2}+\sum_{t\leq 2^k}(\delta^{\tau-t}w_{\tau-t+1})}}.
\end{matrix}
\]
When $\delta=1$, these preference relations are specifically said to follow the \textit{simple sum} criterion.

The formula may seem complicated, but the idea is simple.
The sum is made up of three parts: 
The first term adds up the discounted values of present and near-future payoffs, while the last term adds up those at the end and those that are finite steps away from the end.
The middle term represents a discounted value of the distant future payoffs, which are too far beyond the finite horizon but before the end to be distinguished, and so have been reduced to a single value, $v_{\tau/2}$, just as a vanishing point.

The reason why the values $v_{\tau/2}$ are given an extremely high weight is that all the weights carried by those points that are indiscernible from the point $\tau/2$ are concentrated in this single point. 
All the points far beyond the finite horizon and behind the last term are degenerated to this point, and so this single value $v_{\tau/2}$ carries the weight of all of them.
It also implies that $(v_t)$ and $(w_{t})$ are considered to be equal if $v_t =w_t$ for all $t\in\widehat{\choice{}{}}$.
Since the periods outside $\widehat{\choice{}{}}$ cannot be distinguished by the perspective view so that it is assumed that $(v_t)$ and $(w_t)$ satisfy $v_t=v_{t'}$ and $w_t=w_{t'}$ for all $t\ \widehat{=}\ t'$, or they are consistent with the perspective view as in Definition \ref{consistent_perspective}.

The reason why the values are given by their monads is also important.
These limits may be rational cuts, since they do not have rational representations in general.
For example, consider the sequence $(v_t)$ where $v_t=1$ for all $t\in\FN$ and $v_t=0$ otherwise.
The value of the discounted sum of $(v_i)$ is given by $\frac{1-1/\FN}{1-\delta}=\frac{1}{1-\delta}-\frac{1}{\FN}$, which is a rational cut and is approximated by $\lim_{k\in\FN}\mon{}{\frac{1-\delta^k}{1-\delta}}=\frac{1}{1-\delta}$\footnote{
This is the essence of Kalina measure. See Sakahara and Sato \cite{kalina_measure} for details.} which is real.

The immediate example that satisfies this criterion is a real-valued payoff function $\mathscr{U}_d^{\tau}:\widehat{\mathbf{H}}^\tau\rightarrow R^+$ defined as
\begin{eqnarray*}
  \mathscr{U}_d^{\tau}(\mathbf{h})
  \ =\ \lim_{k\in\FN}\mon{}{
    \begin{matrix}
      \sum_{t\leq 2^k} \Bigl(\delta^{t-1}U(\mathbf{h}(t))\Bigr)\hfill\\
      +(\tau-2^{k+1})\delta^{\tau/2-1}U(\mathbf{h}(\tau/2))\hspace{.3cm}\\
      \hspace{1.5cm}
      +\sum_{t\leq 2^k} \Bigl(\delta^{\tau-t}U(\mathbf{h}(\tau-t+1))\Bigr)
    \end{matrix}
  }
\end{eqnarray*}
where $U(h)\in\FQ$ for all $h\in Z$.
This function is said to be a \textit{discounted sum payoff function} and, separately, it is said to be a \textit{simple sum payoff function}, when $\delta=1$, and is denoted $\mathscr{U}_s^\tau$ of the game $\widehat{\mathbf{\Gamma}}^\tau$.

The condition, $U(h)\in\FQ$, imposed on the values of the payoffs at the terminal histories is essential.
Without this constraint, $\mathscr{U}_d^{\tau}$ cannot satisfy huge transitivity.
For example, suppose $U(h)=\frac{1}{\tau}$ and $U(j)=\frac{2}{\tau}$ for some huge $\tau$ and $\delta=1$, and define a whole history $\mathbf{h}_t={h^{\tau-t}}^{\frown} j^{t}$ for each $t\in \{0,\ldots,\tau\}$.
Then, $\mathscr{U}_s^{\tau}(\mathbf{h}_t)\geq\mathscr{U}_s^\tau(\mathbf{h}_{t+1})$ holds for all $t\in\{0,\ldots,\tau\}$.
However, $\mathscr{U}_s^{\tau}(\mathbf{h}_0)\not\geq\mathscr{U}_s^{\tau}(\mathbf{h}_\tau)$ 
since $\mathscr{U}_s^{\tau}(\mathbf{h}_0)=1$ and $\mathscr{U}_s^{\tau}(\mathbf{h}_\tau)=2$. 

\begin{lemma}\label{ssws}
  A preference relation $\succsim^\tau$ following the discounted sum criterion satisfies weak separability and huge transitivity.
\end{lemma}

\begin{proof}
  Suppose that ${j},{j}'\in{C}$ satisfy ${j}\succsim{j}'$.
  It implies that ${U}({j})\geq {U}({j}')$ and hence $\mathscr{U}_{d}^\tau(\mathbf{h}^\frown({j})^\frown\mathbf{h}')\geq\mathscr{U}_{d}^\tau(\mathbf{h}^\frown({j}')^\frown\mathbf{h}')$.
  So that $\mathbf{h}^\frown({j})^\frown\mathbf{h}'\succsim^\tau\mathbf{h}^\frown({j})'^\frown\mathbf{h}'$ holds.
  
  To prove the second claim, suppose that a chain $(\mathbf{h}_k)_{k\in\{1,\ldots,\kappa\}}$ of $\succsim$ satisfies the relation $\mathbf{h}_\ell\prec^\tau\mathbf{h}_m$ for some $\ell,m\in\{1,\ldots,\kappa\}$ with $\ell<m$.  
  To satisfy $\mathscr{U}_d^\tau(\mathbf{h}_k)\geq\mathscr{U}_d^\tau(\mathbf{h}_{k+1})$ for all $k\in\{\ell,\ldots,m\}$ and $\mathscr{U}_d^\tau(\mathbf{h}_\ell)<\mathscr{U}_d^\tau(\mathbf{h}_{m})$,
  three conditions must be satisfied.
  (1) $m-\ell$ is huge, (2) $\delta<1$, and (3) the number of $k\in\{\ell,\ldots,m-1\}$ satisfying $\mathscr{U}_d^\tau(\mathbf{h}_k)>\mathscr{U}_d^\tau(\mathbf{h}_{k+1})$ is finite.

  Suppose that $\mathscr{U}_d^\tau(\mathbf{h}_k)=\mathscr{U}_d^\tau(\mathbf{h}_{k+1})$ holds for all $k\in\{\ell,\ldots, m-1\}$.
  This implies that $\mathbf{h}_k$ and $\mathbf{h}_{k+1}$ can only have different payoffs after huge periods have elapsed, since $\delta,U(h)\in\FQ$.
  Formally, it must be satisfied that
  \[
  U(\mathbf{h}_{k}(t))=U(\mathbf{h}_{k+1}(t))\qquad\text{ for all } t\in\FN
  \]
  and
  \[
  \lim_{x\in\FN}\mon{}{
    \begin{matrix}
    \left(\tau-2^{x+1}\right)\delta^{\tau/2-1} \left(U(\mathbf{h}_{k}(\tau/2))-U(\mathbf{h}_{k+1}(\tau/2))\right)\hfill\\
    \hspace{.5cm}+\sum_{t\leq 2^x}\delta^{{\tau-t}}\left(U(\mathbf{h}_{k}(\tau-t+1))-U(\mathbf{h}_{k+1}(\tau-t+1))\right)
    \end{matrix}
  }
  \ =\ \mon{}{0}.
  \]
  This implies that $\lim_{x\in\FN}\mon{}{(\tau-2^{x+1})\delta^{\tau/2-1}}\in\mon{}{0}$.
  On the other hand, to satisfy $\mathscr{U}_d^\tau(\mathbf{h}_\ell)<\mathscr{U}_d^\tau(\mathbf{h}_m)$ it is necessary to satisfy
  \[
  \mathscr{U}_d^\tau(\mathbf{h}_m) - \mathscr{U}_d^\tau(\mathbf{h}_\ell)
  \ =\ \lim_{x\in\FN}\mon{}{\left(\tau-2^{x+1}\right)\delta^{\tau/2-1} \underline{q}}
    \ >\ \mon{}{0}
  \]
  where $\underline{q} = U(\mathbf{h}_{m}(\tau/2))-U(\mathbf{h}_{\ell}(\tau/2))$.
  However, since $U(h)\in\FQ$ for all $h\in Z$ and $\lim_{x\in\FN}\mon{}{(\tau-2^{x+1})\delta^{\tau/2-1}}\in\mon{}{0}$, this is impossible.
\end{proof}

Because the discount factor undervalues future payoffs, the gains yielded earlier are preferred to those yielded later, i.e. the terminal whole history $\mathbf{h}^\frown(h,h')^\frown\mathbf{h}'$ 
is preferred to 
$\mathbf{h}^\frown(h',h)^\frown\mathbf{h}'$ if $h\succsim h'$. 
This property is called \textit{the sooner the better} principle.
In fact, the preference relations that follow the discounted sum criterion follow this principle.

\begin{lemma}
    A  preference relation $\succsim^\tau$ that follows the discounted sum criterion satisfies the sooner the better principle.
\end{lemma}

\begin{proof}
  Suppose that ${j},{j}'\in{C}$ satisfy ${j}\succsim{j}'$.
  This implies that ${U}({j})\geq {U}({j}')$ and $\delta{U}({j})\geq \delta{U}({j}')$.
  Thus, ${U}({j})+\delta{U}({j}')\geq {U}({j}')+\delta{U}({j})$ holds since ${U}({j})-{U}({j}')\geq \delta({U}({j})-{U}({j}'))$.
  This implies that $\mathbf{h}^\frown({j},{j}')^\frown\mathbf{h}'\succsim^\tau\mathbf{h}^\frown({j}',{j})^\frown\mathbf{h}'$ holds.
\end{proof}

When the discount factor is 1, the timing of the realisation of each history is no longer an issue. 
This is called the \textit{commutativity} property.
It is easy to see that the preference relations that follow the simple sum criterion always satisfy the commutativity property.

\begin{lemma}
    A preference relation $\succsim^\tau$ that follows the simple sum criterion satisfies commutativity.
\end{lemma}

\begin{proof}
  Since ${U}({j})+{U}({j}')= {U}({j}')+{U}({j})$ holds for all $j,j'\in C$,
  the relation $\mathbf{h}^\frown(j,j')^\frown\mathbf{h}'\sim^\tau\mathbf{h}^\frown(j',j)'^\frown\mathbf{h}'$ holds.
\end{proof}

\begin{ex}[The Prisoner's Dilemma continued from Example \ref{PD}]\label{PD2}
  One might think that a linear preference order would always be obtained under the sooner the better principle, but this is not the case.
  As can be seen from the Hasse diagrams of Player 1's preference relations in Figure \ref{2_pd_discount_pref}, some parts remain unordered.
  
  \begin{figure}[htbp]
    \[
    \centerline{
      \xymatrix@R=14pt@C=10pt@M=2pt{
        && \text{\tiny$(\text{CS},\text{CS})$}
        \ar@{-}[d]
        \\
        && \text{\tiny$(\text{CS},\text{SS})$}
        \ar@{-}[d]
        \ar@{=>}@<2ex>[d]^{\text{\tiny\  as $\delta\rightarrow 1$}}
        \\
        && \text{\tiny$(\text{SS},\text{CS})$}
        \ar@{-}[d]
        \ar@{-}[rrd]
        \\
        && \text{\tiny$(\text{CS},\text{CC})$}
        \ar@{=>}@<2ex>[d]^{\text{\tiny\  as $\delta\rightarrow 1$}}
        \ar@{-}[lld]
        \ar@{-}[d]
        && \text{\tiny$(\text{SS},\text{SS})$}
        \ar@{-}[d]
        \\
        \text{\tiny$(\text{CS},\text{SC})$}
        \ar@{=>}@<-2ex>[d]_{\text{\tiny\  as $\delta\rightarrow 1$}}
        \ar@{-}[d]
        \ar@{-}[rrd]
        && \text{\tiny$(\text{CC},\text{CS})$}
        \ar@{-}[lld]
        \ar@{-}[rrd]
        && \text{\tiny$(\text{SS},\text{CC})$}
        \ar@{=>}@<2ex>[d]^{\text{\tiny\  as $\delta\rightarrow 1$}}
        \ar@{-}[d]
        \ar@{-}[lld]
        \\
        \text{\tiny$(\text{SC},\text{CS})$}
        \ar@{-}[rrd]
        && \text{\tiny$(\text{SS},\text{SC})$}
        \ar@{=>}@<2ex>[d]^{\text{\tiny\  as $\delta\rightarrow 1$}}
        \ar@{-}[d]
        && \text{\tiny$(\text{CC},\text{SS})$}
        \ar@{-}[d]
        \\
        && \text{\tiny$(\text{SC},\text{SS})$}
        \ar@{-}[d]
        && \text{\tiny$(\text{CC},\text{CC})$}
        \ar@{-}[lld]
        \\
        && \text{\tiny$(\text{CC},\text{SC})$}
        \ar@{=>}@<2ex>[d]^{\text{\tiny\  as $\delta\rightarrow 1$}}
        \ar@{-}[d]
        \\
        && \text{\tiny$(\text{SC},\text{CC})$}
        \ar@{-}[d]
        \\
        && \text{\tiny$(\text{SC},\text{SC})$}
        \\
        \save[]+<2cm,0cm>*\txt<14pc>{\scriptsize (a) Preference relations under discounted sum criterion}\restore
      }
      \hspace{.4cm}
      \xymatrix@R=14pt@C=-8pt@M=2pt{
        && \text{\tiny$(\text{CS},\text{CS})$}
        \ar@{-}[d]
        \\
        && \text{\tiny$(\text{CS},\text{CS})$, \tiny$(\text{SS},\text{CS})$}
        \ar@{-}[d]
        \ar@{-}[rrd]
        \\
        && \text{\tiny$(\text{CS},\text{CC})$, \tiny$(\text{CC},\text{CS})$}
        \ar@{-}[lld]
        \ar@{-}[rrd]
        && \text{\tiny$(\text{SS},\text{SS})$}
        \ar@{-}[d]
        \\
        \text{\tiny$(\text{CS},\text{SC})$, \tiny$(\text{SC},\text{CS})$}
        \ar@{-}[rrd]
        && 
        && \text{\tiny$(\text{SS},\text{CC})$, \tiny$(\text{CC},\text{SS})$}
        \ar@{-}[d]
        \ar@{-}[lld]
        \\
        && \text{\tiny$(\text{SS},\text{SC})$, \tiny$(\text{SC},\text{SS})$}
        \ar@{-}[d]
        && \text{\tiny$(\text{CC},\text{CC})$}
        \ar@{-}[lld]
        \\
        && \text{\tiny$(\text{CC},\text{SC})$, \tiny$(\text{SC},\text{CC})$}
        \ar@{-}[d]
        \\
        && \text{\tiny$(\text{SC},\text{SC})$}
        \\
        \save[]+<2cm,0cm>*\txt<14pc>{\scriptsize (b) Preference relation under simple sum criterion}\restore
    }  }  
    \]
    
    \caption{Hasse diagrams of preference relations under two criteria. \label{2_pd_discount_pref}}
  \end{figure}
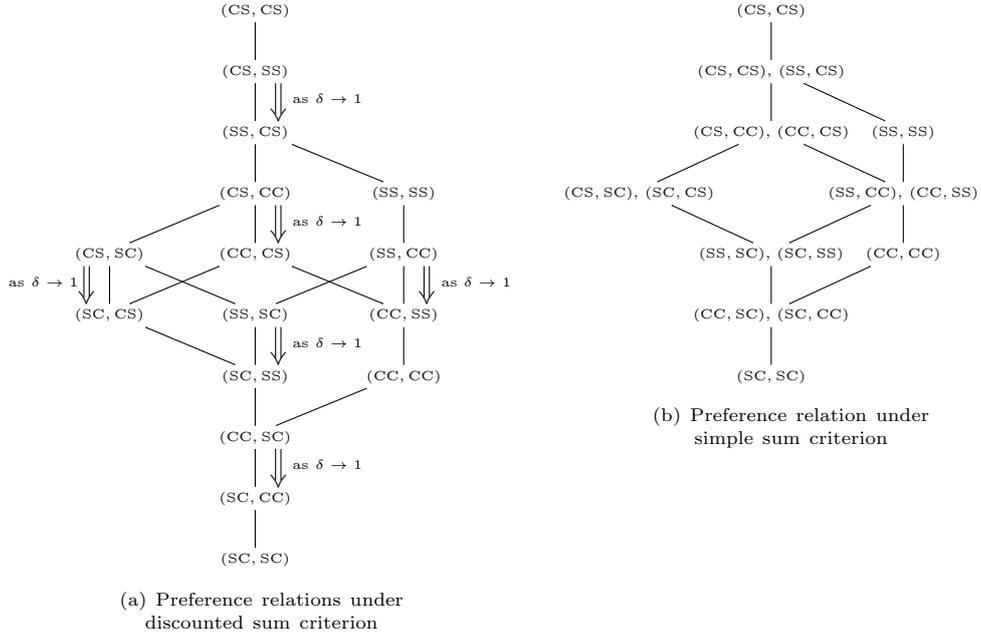

  In fact, the order between $(\text{CS, SC})$ and $(\text{CC, CS})$ varies depending on the level at which each payoff and discount factor is set.
  If they are given as $U_1(\text{CS})=0$,  $U_1(\text{SC})=-5$,  $U_1(\text{CC})=-3$,  $U_1(\text{SS})=-1$ and $\delta=1/5$, the preference relation of the game $\mathbf{\Gamma}^2$ is given as $(\text{CS, SC}) \succsim_1 (\text{CC, CS})$ since $\mathscr{U}_{d\, 1}^2((\text{CS, SC}))=-1$ and $\mathscr{U}_{d\, 1}^2((\text{CC, CS}))=-3$, while the relation is reversed to $(\text{CC, CS}) \succsim_1 (\text{CS, SC})$ when  $U_1(\text{SC})=-25$, since $\mathscr{U}_{d\,1}^2((\text{CS, SC}))=-5$ and $\mathscr{U}_{d\,1}^2((\text{CC, CS}))=-3$.

  The order also changes if the discount factor is set to $\delta=1$.
  In fact, in this case, the payoff also changes to $\mathscr{U}_s^2((\text{CS, SC}))=-5$, and $(\text{CC, CS}) \succsim_1 (\text{CS, SC})$ follows, as shown in Figure \ref{2_pd_discount_pref} (b).
\end{ex}

When the discount factor is 1 and every payoff in every period is positive and finite rational, the payoffs over all periods reach huge values.
This makes the payoffs yielded for the whole histories of huge-repeated games indiscernible from each other and creates space for a new kind of subgame perfect equilibria, characteristic of the simple sum criterion, to emerge.
The next proposition confirms the existence of such equilibria.

\begin{prop}\label{simple_sum}
  Suppose that there exists a connected terminal history $h\in C$ where the payoff of $h$ is positive and finite rational for all players consisting only of core players, and the payoff function is given by a simple sum.
  Then, the strategy profile $\mathbf{s}^s$ that plays $a=s^{*}_{{P}(\mathbf{h}(|\mathbf{h}|))}(\mathbf{h}(|\mathbf{h}|))$ for any whole history $\mathbf{h}$ that satisfies $\mathbf{h}(|\mathbf{h}|)\not\subset{h}$, or $\tau-|\mathbf{h}|\in \FN$ and contains only finite $h$, and takes action $a$ that satisfies $\mathbf{h}(|\mathbf{h}|)^\frown (a)\subseteq h$  otherwise,
  is a subgame perfect equilibrium.
\end{prop}

\begin{proof}
  Let the payoff of the player $i$ at an arbitrarily chosen whole history $\mathbf{h}$ of  $\mathbf{\Gamma}^\tau$, in which $i=\mathbf{P}^\tau(\mathbf{h})$, be given by $U_i(h)=y_i>0$.

  Provided that $\mathbf{h}$ satisfies $\mathbf{h}(|\mathbf{h}|)\subseteq h$, or the history of the constituent game $\Gamma$ of the $|\mathbf{h}|$-th period is consistent with $h$.
  First, suppose $\mathbf{h}$ is a whole history far from the end,
  then following $\mathbf{s}^s_{i}|_{\mathbf{h}}$ 
  yields $\mon{}{y_i\cdot (\tau-|\mathbf{h}|)}$, which is infinite.
  Second, if $\mathbf{h}$ includes huge $h$, then every action yields a huge payoff. 
  This implies that $\mathbf{s}^s$ achieves the infinite real payoff, so that all players have no incentive to deviate.
  Otherwise, if the whole history $\mathbf{h}$ satisfies $\mathbf{h}(|\mathbf{h}|)\not\subset h$, then the strategy $s^{*\tau}_{i}|_{\mathbf{h}}$ is guaranteed to be a subgame perfect equilibrium of $\Gamma(\mathbf{h})$ by Proposition \ref{ext} and Lemma \ref{ssws}, and hence $\mathbf{s}^s$ is a subgame perfect equilibrium of $\mathbf{\Gamma}^\tau$.   
\end{proof}

The statement may seem to be complicated.
But the intuitive meaning is simple.
It means that if there is an opportunity to make huge gains, you should keep playing, regardless of how your opponents are playing.

The result points to the existence of new equilibria that have not been found before, but which seem intuitively obvious.
Let us examine these with three simple examples.

\begin{ex}[Centipede Games continued from Example \ref{dc}]\label{dc_ss}
Figure \ref{con_cp_ne} shows a numerical example of constituent games of centipede games, whose subgame perfect equilibrium $s^*$ is uniquely given by $s^*_1(\emptyset)=\mathrm{D}$ and $s^*_2(\mathrm{R})=\mathrm{d}$.

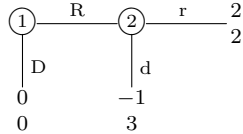
\begin{figure}[htbp]
  \[
  \centerline{
    \xymatrix@R=14pt@C=30pt@M=-.10pt{
      \raise0.2ex\hbox{\textcircled{\scriptsize{1}}}
      \ar@{-}[r]^{\mathrm{R}}
      \ar@{-}[d]^{\mathrm{D}}
      &
      \raise0.2ex\hbox{\textcircled{\scriptsize{2}}}
      \ar@{-}[r]^{\mathrm{r}}
      \ar@{-}[d]^{\mathrm{d}}
      &
      \,\text{\footnotesize$\begin{matrix}2\\2
        \end{matrix}$}
      \\ 
      \text{\footnotesize$\begin{matrix}0\\0
        \end{matrix}$}
      & 
      \text{\footnotesize$\begin{matrix}-1\\3
        \end{matrix}$}
      & 
    }}
  \]
  \caption{A tree of the constituent game of a centipede game.}\label{con_cp_ne}
\end{figure}

The game satisfies the dynamic consistency condition if its payoffs are given by the simple sum criterion.
To see if this condition is indeed met, let us examine the tree of the game repeated up to 5 periods with simple sum payoffs, as shown in Figure \ref{5_cp}.
\begin{figure}[htbp]
  \[
  \centerline{
    \xymatrix@R=14pt@C=20pt@M=-.10pt{
      \text{\scriptsize{1}}
      &&
      \text{\scriptsize{2}}
      &&
      \text{\scriptsize{3}}
      &&
      \text{\scriptsize{4}}
      &&
      \text{\scriptsize{5}}
      &&
      \\
      \raise0.2ex\hbox{\textcircled{\scriptsize{1}}}\ar@{-}[r]^{\mathrm{R}}\ar@{-}[d]^{\mathrm{D}}
      & 
      \raise0.2ex\hbox{\textcircled{\scriptsize{2}}}\ar@{-}[r]^{\mathrm{r}}\ar@{-}[d]^{\mathrm{d}}
      & 
      \raise0.2ex\hbox{\textcircled{\scriptsize{1}}}\ar@{-}[r]^{\mathrm{R}}\ar@{-}[d]^{\mathrm{D}}
      & 
      \raise0.2ex\hbox{\textcircled{\scriptsize{2}}}\ar@{-}[r]^{\mathrm{r}}\ar@{-}[d]^{\mathrm{d}}
      & 
      \raise0.2ex\hbox{\textcircled{\scriptsize{1}}}\ar@{-}[r]^{\mathrm{R}}\ar@{-}[d]^{\mathrm{D}}
      &
      \raise0.2ex\hbox{\textcircled{\scriptsize{2}}}\ar@{-}[r]^{\mathrm{r}}\ar@{-}[d]^{\mathrm{d}}
      &  
      \raise0.2ex\hbox{\textcircled{\scriptsize{1}}}\ar@{-}[r]^{\mathrm{R}}\ar@{-}[d]^{\mathrm{D}}
      & 
      \raise0.2ex\hbox{\textcircled{\scriptsize{2}}}\ar@{-}[r]^{\mathrm{r}}\ar@{-}[d]^{\mathrm{d}}
      & 
      \raise0.2ex\hbox{\textcircled{\scriptsize{1}}}\ar@{-}[r]^{\mathrm{R}}\ar@{-}[d]^{\mathrm{D}}
      &
      \raise0.2ex\hbox{\textcircled{\scriptsize{2}}}\ar@{-}[r]^{\mathrm{r}}\ar@{-}[d]^{\mathrm{d}}
      &
      \,\text{\footnotesize$\begin{matrix}10\\10
        \end{matrix}$}
      \\ 
      \text{\footnotesize$\begin{matrix}0\\0
        \end{matrix}$}
      & 
      \text{\footnotesize$\begin{matrix}-1\\3
        \end{matrix}$}
      &
      \text{\footnotesize$\begin{matrix}2\\2
        \end{matrix}$}
      & 
      \text{\footnotesize$\begin{matrix}1\\5
        \end{matrix}$}
      & 
      \text{\footnotesize$\begin{matrix}4\\4
        \end{matrix}$}
      & 
      \text{\footnotesize$\begin{matrix}3\\7
        \end{matrix}$}
      &
      \text{\footnotesize$\begin{matrix}6\\6
        \end{matrix}$}
      & 
      \text{\footnotesize$\begin{matrix}5\\9
        \end{matrix}$}
      & 
      \text{\footnotesize$\begin{matrix}8\\8
        \end{matrix}$}
      & 
      \text{\footnotesize$\begin{matrix}7\\11
        \end{matrix}$}
  }}
  \]
  \caption{A tree of the 5-repeated centipede game with simply summed payoffs.\label{5_cp}}
\end{figure}

It is clear that the game satisfies dynamic consistency, $(h)^\frown(O(s^*))\sim^5_i(h)$ for all $i\in I$ and $h\in C$, since $\mathscr{U}^5_i\left((\mathrm{Rr,D})\right)=\mathscr{U}^5_i\left((\mathrm{Rr})\right)$ and hence $\mathrm{(Rr,D)}\sim^5_i\mathrm{(Rr)}$ holds.
This implies that the strategy $s^{*5}$, which repeats the strategy $s^*$ 5 times, is a subgame perfect equilibrium, as already confirmed in Proposition \ref{ext}.
Indeed the unique subgame perfect equilibrium is given by ${s}^{*5}_1(\mathbf{h})=\mathrm{D}$ for all $\mathbf{P}^5(\mathbf{h})=1$ and ${s}^{*5}_2(\mathbf{h})=\mathrm{d}$ for all $\mathbf{P}^5(\mathbf{h})=2$.

As noted above, the result is confusing, because both players benefit from cooperating to continue the game into the final period. 
However, the subgame perfect equilibrium of the game excludes this kind of cooperation at all.

The situation changes when the game is extended to huge periods.
Extending the game beyond the finite horizon to huge $\tau$ periods leads to new subgame perfect equilibria to emerge as indicated in Proposition \ref{simple_sum}.
Figure \ref{T_cp} shows the game extended to $\tau$ periods in the perspective view where $\tau\in N\setminus\FN$.
\begin{figure}[htbp]
  \[
  \centerline{
    \xymatrix@=17pt@M=-.10pt{
      \text{\scriptsize{1}}
      &&
      \text{\scriptsize{2}}
      &&&&
      \text{\scriptsize{$\tau$}}
      &
      \\
      \raise0.2ex\hbox{\scriptsize\textcircled{\scriptsize{1}}}
      \ar@{-}[r]^{\text{\scriptsize R}}
      \ar@{-}[d]^{\text{\scriptsize D}}
      & 
      \raise0.2ex\hbox{\scriptsize\textcircled{\scriptsize{2}}}
      \ar@{-}[r]^{\text{\scriptsize r}}
      \ar@{-}[d]^{\text{\scriptsize d}}
      & 
      \raise0.2ex\hbox{\scriptsize\textcircled{\scriptsize{1}}}
      \ar@{-}[r]^{\text{\scriptsize R}}
      \ar@{-}[d]^{\text{\scriptsize D}}
      & 
      \raise0.2ex\hbox{\scriptsize\textcircled{\scriptsize{2}}}
      \ar@{-}[r]^(.6){\text{\scriptsize r}}
      \ar@{-}[d]^{\text{\scriptsize d}}
      && 
      \cdots
      &
      \raise0.2ex\hbox{\scriptsize\textcircled{\scriptsize{1}}}
      \ar@{-}[r]^{\text{\scriptsize R}}
      \ar@{-}[d]^{\text{\scriptsize D}}
      &
      \raise0.2ex\hbox{\scriptsize\textcircled{\scriptsize{2}}}
      \ar@{-}[r]^(.22){\text{\scriptsize r}}
      \ar@{-}[d]^{\text{\scriptsize d}}
      &
      \text{\footnotesize$\begin{matrix}
          \mon{}{2\tau} \\
          \mon{}{2\tau} 
        \end{matrix}$}
      \\ 
      \text{\footnotesize$\begin{matrix}
          0\\
          0 
        \end{matrix}$}
      & 
      \text{\footnotesize$\begin{matrix}
          -1\\
          3
        \end{matrix}$}
      & 
      \text{\footnotesize$\begin{matrix}
          2\\
          2 
        \end{matrix}$}
      & 
      \text{\footnotesize$\begin{matrix}
          1\\
          5 
        \end{matrix}$}
      &&&
      \text{\footnotesize$\begin{matrix}
          \mon{}{2\tau-2}\\
          \mon{}{2\tau-2} 
        \end{matrix}$}
      & 
      \text{\footnotesize$\begin{matrix}
          \mon{}{2\tau-3}\\
          \mon{}{2\tau+1} 
        \end{matrix}$}
  }}
  \]
  \caption{A tree of the $\tau$-repeated centipede game.\label{T_cp}}
\end{figure}

The payoffs yielded at $\tau$-period are given by
\[
\begin{matrix}
  \mathscr{U}_{s\,1}(\mathbf{Rr}^{\tau-1\frown}(\text{D}))
  \ =\ \mon{}{2\tau-2}\ =\ \infty\\[1ex]
  \mathscr{U}_{s\,2}(\mathbf{Rr}^{\tau-1\frown}(\text{D}))
  \ =\ \mon{}{2\tau-2}\ =\ \infty\\[2ex]
  \mathscr{U}_{s\,1}(\mathbf{Rr}^{\tau-1\frown}(\text{Rd})
  \ =\ \mon{}{2\tau-3}\ =\ \infty\\[1ex]
  \mathscr{U}_{s\,2}(\mathbf{Rr}^{\tau-1\frown}(\text{Rd})
  \ =\ \mon{}{2\tau+1}\ =\ \infty\\[2ex]
  \mathscr{U}_{s\,1}(\mathbf{Rr}^{\tau})\hfill
  \ =\ \mon{}{2\tau}\ =\ \infty\\[1ex]
  \mathscr{U}_{s\,2}(\mathbf{Rr}^{\tau})\hfill
  \ =\ \mon{}{2\tau}\ =\ \infty
\end{matrix}
\]

Since each payoff is infinite, for any nonterminal whole history $\mathbf{h}$ where $|\mathbf{h}|=\tau$, both actions R and D for player 1 or r and d for player 2 at the last game are indifferent.
So that $\mathbf{s}_1^s(\mathbf{h})=\mathrm{R}$ for all $\mathbf{P}^\tau(\mathbf{h})=1$ and $\mathbf{s}_2^s(\mathbf{h})=\mathrm{r}$ for all $\mathbf{P}^\tau(\mathbf{h})=2$ for any $\mathbf{h}\in \mathbf{H}^\tau\setminus\mathbf{Z}^\tau$ also emerges as a subgame perfect equilibrium.

This equilibrium may correspond to our natural understanding of the game.
The reason why this profile is the equilibrium is simple.
Since the payoffs from these strategy profiles ${\mathbf{s}^s}$ remain huge, neither player has an incentive to deviate, so that the profile ensures its stability.
\end{ex}

\begin{ex}[Ultra Long-Term Investment continued from Example \ref{ulti}]\label{ulti_ss}
  A numerical example that illustrates interesting aspects of the decision problem under the simple sum criterion is provided by an ultra long-term investment.

  A player is repeatedly confronted with the same decision problems whether to continue investing, I, or not to invest, N.
  If investments are made in every single period, then enormous gains, $\tau$, are yielded.
  On the other hand, if no investment is made even once, the payoff is less than 0 because all the investments already made are all lost.
  
  The player's payoff function  is given by
  \[
  \begin{matrix}
    \mathscr{U}_s({I}^\tau) & = & \mon{}{\tau}\text{\hspace{7.2cm}}\\
    \mathscr{U}_s(\mathbf{h})
    & = & -\lim_{k\in\FN}\mon{}{
      \begin{matrix}
        \sum_{t\in 2^k}[\mathbf{h}(t)=\text{I}]\text{\hspace{3cm}}\\
        +(\tau-2^{k+1})[\mathbf{h}(\tau/2)=\text{I}]\\
        \text{\hspace{2cm}}+\sum_{t\in 2^k}[\mathbf{h}(t)=\text{I}]
      \end{matrix}
        }
  \end{matrix}
  \]
  where $\mathbf{h}$ contains some histories in which no investment is made.

  One of the subgame perfect equilibria characteristic of the perspective view is given as Invest for the whole history $\mathbf{h}(t)=\text{I}$ for $t\in\{1,\ldots,\tau\}$ if $\mathbf{h}(t-k)=\text{I}$ for all $k\in\{1,\ldots,t-1\}$, and not otherwise.
  As a result, investment takes place on the equilibrium path.
\end{ex}

\begin{ex}[Lifestyle Disease continued from Example \ref{addiction}]\label{addiction_ss}
  Another example that illustrates the peculiarities of the decision problem is given by lifestyle disease.

  A player is faced with the problem of whether to eat a food that will cause some health problems in the distant future, E (eat), or avoid it, A (avoid).
  Eating the food provides some pleasure, which is assumed to be equal to 1, but it also causes a health problem that appears very late, which is equal to $-2$, after more than finitely many periods have elapsed, i.e. more than $k$-periods have elapsed for all $k\in \FN$.

  The player's payoff function is given by
  \[
  \begin{matrix}
    \mathscr{U}_s(\mathbf{h})
    & = & \lim_{k\in\FN}\mon{}{
      \begin{matrix}
        -\sum_{t\in 2^k}[\mathbf{h}(t)=\text{U}]\text{\hspace{3cm}}\\
        +(\tau-2^{k+1}-2(\tau-\alpha))[\mathbf{h}(\tau/2)=\text{U}]\\
        \text{\hspace{3cm}}+\sum_{t\in 2^k}[\mathbf{h}(t)=\text{U}]
      \end{matrix}
    }
  \end{matrix},
  \]
  where $\mathbf{h}$ is an arbitrary given whole history.

  One of the subgame perfect equilibria characteristic of the perspective view is given as Avoid until the whole history $\mathbf{h}$ satisfies $|\mathbf{h}|<\tau-k$ for all $k\in\FN$, and Eat otherwise.
  The player will only eat if there are finite periods left, since the player will be dead at period $\tau$ and there is no risk of suffering severe health problems.

This result may seem obvious. However, this is not the case when the con- ventional decision-making framework is employed, as will become clear later on.
\end{ex}

So far, the case where the payoff at each terminal history is set to a positive finite rational has been considered.
Once this condition is violated, the existence of these equilibria is not guaranteed.
Let us examine one case, the Prisoner's Dilemma problem, to see how it fails.

\begin{ex}[The Prisoner's Dilemma continued from Example \ref{PD}]
  Consider the Prisoner's Dilemma game, where the payoff is given by Table \ref{one_ex}.
  \begin{table}[htbp]
    \begin{center}
      {\renewcommand\arraystretch{1.5}
        \begin{tabular}{r|c|c|}
          \multicolumn{1}{c}{} 
          & \multicolumn{1}{c}{\footnotesize Silence} 
          & \multicolumn{1}{c}{\footnotesize Confess} \\
          \cline{2-3}
          \footnotesize Silence & 3,3 & 0,4\\
          \cline{2-3}
          \footnotesize Confess & 4,0 & 1,1\\
          \cline{2-3}
        \end{tabular}
      }
    \end{center}
    \caption{One example of the Prisoner's Dilemma Game}\label{one_ex}
  \end{table}
  
  As can be seen from the matrix, the pair of strategies $s_1=s_2=$ S satisfies the condition of Proposition \ref{simple_sum}, since S yields positive payoffs for both players and S is connected terminal.
  This implies that the profile of strategies $\mathbf{s}^s_i({\mathbf{h}})=\text{C}$ for all $\mathbf{h}$ satisfying $\tau-|\mathbf{h}|\in\FN$ and containing only finite S, and $\mathbf{s}^s_i({\mathbf{h}})=\text{S}$ otherwise 
  for all $t\in\{1,\ldots,\tau\}$ and all $i\in\{1,2\}$ is a subgame perfect equilibrium of this game.
  \[
  \mathbf{s}^s({\mathbf{h}})\ =\
  \begin{cases}
    \text{C} & \text{ if } (\tau-|\mathbf{h}|\in\FN)\wedge(|\{t\in|\mathbf{h}|\,;\, \mathbf{h}(t)=\text{S}\}|\in\FN)\\
    \text{S} & \text{ otherwise}.
  \end{cases}
  \]
  
  The equilibrium path $\mathbf{s}^s$ is given as $\text{S}^\tau$.

  However, the situation changes when a payoff matrix is replaced by the following:
  \begin{table}[htbp]
    \begin{center}
      {\renewcommand\arraystretch{1.5}
        \begin{tabular}{r|c|c|}
          \multicolumn{1}{c}{} 
          & \multicolumn{1}{c}{\footnotesize Silence} 
          & \multicolumn{1}{c}{\footnotesize Confess} \\
          \cline{2-3}
          \footnotesize Silence & $-1$, $-1$ & $-4$, 0\\
          \cline{2-3}
          \footnotesize Confess & 0, $-4$ & $-3$, $-3$\\
          \cline{2-3}
        \end{tabular}
      }
    \end{center}
    \caption{Another example of the Prisoner's Dilemma Game}\label{another}
  \end{table}

  Note that all the elements of the matrix in Table \ref{another} have been reduced by 4 compared to those in Table \ref{one_ex}.
  As expected, the differences in payoffs, and thus the preferred actions, remain unchanged.
%
  
  However, the pair of strategies $s_1=s_2=$ S does not satisfy the condition of Proposition \ref{simple_sum}, since S yields only negative payoffs for both players.
  This implies that the profile of strategies $(\mathbf{s}^s_i)$ constructed above is no longer a subgame perfect equilibrium.
  In fact, both players can increase their payoffs by changing their strategy from $\mathbf{s}^s_i$ to $s^{*\tau}_i$, since $\mathscr{U}_{s\,i}^\tau(\mathbf{s}^s)=\mon{}{-\tau}=-\infty$ while $\mathscr{U}_{s\, i}^\tau({s}^{*\tau}_i,\mathbf{s}^{s}_{-i})=0$.
\end{ex}

In contrast, under the strict discounted sum criterion, $\mathbf{s}^s$ cannot be a subgame perfect equilibrium.
This is because, as the proof of Lemma \ref{ssws} shows, this criterion ignores all the payoffs yielded after huge periods of time have elapsed.
Instead, the strategy that allows any paths to be realised after huge periods emerges as the new equilibrium.

\begin{prop}\label{discount}
  If $\delta < 1$, $\mathbf{s}^s$ is no longer guaranteed to be a subgame perfect equilibrium unless $\mathbf{s}^s=s^{*\tau}$.
  Instead, the strategy profile $\mathbf{s}^d$, which assigns any action $a\in A(\mathbf{h}(|\mathbf{h}|))$ if $\mathbf{h}$ is not a near-future history, i.e. $|\mathbf{h}|\notin\FN$, and plays $a=s^*_{P(\mathbf{h}(|\mathbf{h}|))}({\mathbf{h}(|\mathbf{h}|)})$ otherwise, is a subgame perfect equilibrium.
\end{prop}

\begin{proof}
  Since $\tau\cdot \delta^{\tau/2-1}$ is indiscernible from 0, the discounted sum of the payoff given by $\mathbf{s}^d$ in a distant future, i.e. $(\tau-2^{k+1})\delta^{\tau/2-1} U_i(\mathbf{O}(\mathbf{s}^d|_{\mathbf{h}})(\tau/2))$ is indiscernible from $0$ for all $k\in\FN$ for all  player $i\in I$.
  This implies that it does not matter what action is taken at $\mathbf{h}$.
  It also implies that it matters at any whole history of near future, i.e. $|\mathbf{h}|\in\FN$, so that $\mathbf{s}^d$ is guaranteed to be  a subgame perfect equilibrium and  $\mathbf{s}^s$ is not.
\end{proof}

\begin{ex}[Centipede Games continued from Example \ref{dc_ss}]
  The subgame perfect equilibrium examined in Example \ref{dc_ss}, in which both players play R or r to continue the game to the end, disappears when the payoff function is given by strict discounting.
  This is because discounting reduces the total payoffs to finite values, making any deviation profitable.
  In fact, deviating from R to D for player 1 and from r to d for player 2 always yields positive gains in every period $k\in\FN$.
  The resulting subgame perfect equilibrium is given by player 1 playing D and 2 playing d in every period.
  The trust that seemed to exist under the equilibrium $\mathbf{s}^s$ has disappeared.
\end{ex}

\begin{ex}[Ultra Long-Term Investment continued from Example \ref{ulti_ss}]
  The subgame perfect equilibrium examined in Example \ref{ulti_ss}, in which an ultra-long-term investment is made, also disappears. 
  This is because discounting reduces future gains to zero, while near future costs remain intact.
  The resulting subgame perfect equilibrium is given by always avoiding investment.
  The long-term perspective that seemed to exist under $\mathbf{s}^s$ has also disappeared.
\end{ex}

\begin{ex}[Lifestyle Disease continued from Example \ref{addiction_ss}]
  The same effect is also observed in the lifestyle disease problem.
  The subgame perfect equilibrium examined in Example \ref{addiction_ss}, in which people always avoid eating the food unless the remaining time is finite, disappears again. 
  This is partly because discounting reduces the future cost of the distant future health problem to zero, while leaving the pleasures available in the near future intact.
  The resulting subgame perfect equilibrium is given by always eating the food. 
  In essence, those who discount and nullify future costs risk their long-term health.
\end{ex}

\subsection{Overtaking}

The profile of preference relations $(\succsim_i^\tau)$ of generalised repeated games is said to follow the \textit{overtaking} criterion of constituent games if the sequence $(v_t)_{t\in\{1,\ldots,\tau\}}$ of payoffs is preferred to the sequence $(w_t)_{t\in\{1,\ldots,\tau\}}$ if and only if
\[
\int_{\widehat{\mathscr{T}}}(v_t-w_t)d\widehat{m}(t) \ \geq \ 0
\]
or, equivalently,
\[
{\lim_{k\in\FN}\mon{}{\sum_{t\leq 2^k}(v_t-w_t)+(\tau-2^k)(v_\tau-w_\tau)+\sum_{t\leq 2^k}(v_{\tau-t+1}-w_{\tau-t+1})}}\geq 0.
\]

It may seem that the preference relations obtained by this criterion are identical to those obtained by the simple sum criterion.
However, they are different.
For example, suppose that $U_i(h)=1$ and $U_i(j)=2$ for the constituent game $\Gamma$ and $\mathbf{h}=(h)_{t\in\{1,\ldots,\tau\}}$ and $\mathbf{j}=(j)_{t\in\{1,\ldots,\tau\}}$.
Then, the payoffs given by the simple sum payoff function are $\mathscr{U}_{s\, i}^\tau(\mathbf{h})=\mon{}{\tau}=\infty=\mon{}{2\tau}=\mathscr{U}_{s\,i}^\tau(\mathbf{j})$, so that $\mathbf{h}\sim_i^\tau\mathbf{j}$ holds.
However, the relation induced by the overtaking criterion is $\mathbf{h}\not\succsim_i^\tau\mathbf{j}$, since $\int_{\widehat{\mathscr{T}}}(v_t-w_t)d\widehat{m} =\mon{}{\tau-2\tau}=-\infty< 0$.
Consequently, the difference allows the preference relations induced by the overtaking criterion to satisfy strict separation.

\begin{lemma}\label{otws}
  A  preference relation $(\succsim^\tau)$ following the overtaking criterion satisfies strict separability and huge transitivity.
\end{lemma}

\begin{proof}
  Suppose that ${j},{j}'\in{C}$ satisfy ${j}\succ{j}'$.
  It implies that ${U}({j})> {U}({j}')$.
  For each pair of whole histories $\mathbf{h}, \mathbf{h}'$ which satisfy $\mathbf{j}=\mathbf{h}^\frown (j)^\frown \mathbf{h}'\in \mathbf{Z}^\tau$ and $\mathbf{j}'=\mathbf{h}^\frown (j')^\frown \mathbf{h}'\in \mathbf{Z}^\tau$, the inequation $\int_{\widehat{\mathscr{T}}}(U(\mathbf{j}(t))-U(\mathbf{j}'(t)))d\widehat{m}(t)=U(j)-U(j')>0$ holds. 
  So that $\mathbf{h}^\frown({j})^\frown\mathbf{h}'\succ^\tau\mathbf{h}^\frown({j})'^\frown\mathbf{h}'$ holds.
  
  To prove the second, it suffices to show that there is no chain $(\mathbf{h}_k)_{k\in\{1,\ldots,\kappa\}}$ of $\succsim_i^\tau$ such that $\int_{\widehat{\mathscr{T}}}(U(\mathbf{h}_{k+1}(t))-U(\mathbf{h}_{k}(t)))d\widehat{m}(t) =0$ for all $k\in\{1,\ldots,\kappa-1\}$ and $\int_{\widehat{\mathscr{T}}}(U(\mathbf{h}_\ell(t))-U(\mathbf{h}_m(t)))d\widehat{m}(t) <0$ for some $\ell>m$.
  Since $U_i(h)\in\FQ$ for all $h\in Z$, the difference between two payoffs becomes 0 only if $\sum_{t\in\widehat{\choice{}{}}}{U}(\mathbf{h}_{k+1}(t))=\sum_{t\in\widehat{\choice{}{}}}{U}(\mathbf{h}_{k}(t))$. 
  This implies that $\mathbf{h}_\ell\succsim_i^\tau\mathbf{h}_m$ holds for all $\ell>m$.
\end{proof}

All of the new subgame perfect equilibria that emerged under the simple sum criterion are realised because of the hugeness of the payoffs that are summed up.
Therefore, it is not surprising that these equilibria disappear when the simple sum criterion is discarded and replaced by the overtaking criterion.
As will be demonstrated below, there is no room for these equilibria to exist under the overtaking criterion.

\begin{prop}\label{nonext}
 Suppose a constituent game $\Gamma$ has a unique subgame perfect equilibrium $s^*$ and a profile of preference relations $(\succsim_i^\tau)$ of a $\tau$-repeated game $\mathbf{\Gamma}^\tau$ following the overtaking criterion satisfies dynamic consistency, then the strategy profile ${s}^{*\tau}$ is the only subgame perfect equilibrium of $\mathbf{\Gamma}^\tau$.
\end{prop}

\begin{proof}
  Since weak separability and huge transitivity is satisfied by Lemma \ref{otws}, each ${s}^{*\tau}$ is a subgame perfect equilibrium of $\mathbf{\Gamma}^\tau$ by Proposition \ref{ext}.

  Let $\mathbf{s}^*$ denote a subgame perfect equilibrium of $\mathbf{\Gamma}^\tau$.
  Since the subgame perfect equilibrium $s^*$ of $\Gamma$ is unique, $\mathbf{s}^*$ must satisfy $\mathbf{s}_{i}^*|_{\mathbf{h}_{\tau}}((h))=s_{i}^*|_{\mathbf{h}_{\tau}(\tau)}(h)$ for all $\mathbf{h}_{\tau}\in {C}^{\tau-1}\times{H}$ and $h\in {H}|_{\mathbf{h}_{\tau}(\tau)}$ where $i=\mathbf{P}^\tau|_{\mathbf{h}_\tau}((h))$.
  
  Suppose $\mathbf{s}^*$ coincides with $s^{*\tau}$ in the subgames after the $k+1$-th period, i.e. $\mathbf{s}^*$ satisfies $\mathbf{s}^*_{i}|_{\mathbf{h}_{k+1}}(\mathbf{j})=s^{*\tau}_{i}|_{\mathbf{h}_{k+1}}(\mathbf{j})$ for all whole histories $\mathbf{h}_{k+1}\in C^k\times H$ 
  and $\mathbf{j}\in \mathbf{H}^{\tau}|_{\mathbf{h}_{k+1}}$ for some $k\in\{1,\ldots,\tau-1\}$ where $i=\mathbf{P}^\tau|_{\mathbf{h}_{k+1}}(\mathbf{j})$.
  Then, the payoff yielded by the strategy $\mathbf{s}^*$ in the subgame $\mathbf{\Gamma}|_{\mathbf{h}_k}$, where $\mathbf{h}_{k}\in{C}^{k-1}\times{H}$, can only have different actions in the subgame of the $k$-th period.
  Since $s^*$ is unique and $(\succsim_i^\tau)$ is strictly separable by Lemma \ref{otws}, the following relation holds for all $s_i$, where $i=P(\mathbf{h}_k(k))$ and $h=O_{\mathbf{h}_k(k)}(s_{i}|_{\mathbf{h}_k(k)},s^*_{-i}|_{\mathbf{h}_k(k)})$, 
  \[
  \mathbf{O}_{\mathbf{h}_k}(s^{*\tau}|_{\mathbf{h}_k})
  \ \succ_{i}^\tau\hspace{-1mm}|_{\mathbf{h}_k}\ 
  (h)^\frown \mathbf{O}_{\mathbf{h}_k^\frown({h})}(s^{*\tau}|_{\mathbf{h}_k^\frown({h})}).
  \]
  It implies that $\mathbf{s}^*$ coincides with $s^{*\tau}$, since $k$ is arbitrarily chosen from $\{1,\ldots,\tau-1\}$.
\end{proof}

\subsection{Limit of Means}

The preference relations $(\succsim_i^\tau)$ of generalised repeated games are said to follow the \textit{limit of means} criterion of constituent games if the sequence $(v_{t})_{t\in\{1,\ldots,\tau\}}$ of payoffs is preferred to the sequence $(w_t)_{t\in\{1,\ldots,\tau\}}$ if and only if
\[
\int_{\circtop{\mathscr{T}}}v_td\circtop{m} \ \geq \ \int_{\circtop{\mathscr{T}}}w_td\circtop{m}
\]
or, equivalently,
\begin{eqnarray*}
{\lim_{k\in\FN}\mon{}{\sum_{0\leq i\leq \min(k,1)}\frac{1}{2^{k}}v_{t_{(i,0)}}+\sum_{2\leq i\leq k}\sum_{2\leq j\leq 2^{i-1}}\frac{1}{2^{k-1}}v_{t_{(i,j-2)}}}}\\
&\hspace{-9.6cm} \geq 
\displaystyle{{\lim_{k\in\FN}\mon{}{\sum_{0\leq i\leq \min(k,1)}\frac{1}{2^{k}}w_{t_{(i,0)}}+\sum_{2\leq i\leq k}\sum_{2\leq j\leq 2^{i-1}}\frac{1}{2^{k-1}}w_{t_{(i,j-2)}}}}}.
\end{eqnarray*}

Note that these values are computed only for the elements of $\circtop{\choice{}{}}$, the choice class of $\{1,\ldots,\tau\}$ by the bird's eye view.
This implies that $(v_t)$ and $(w_{t})$ are considered to be equal if $v_t =w_t$ for all $t\in\circtop{\choice{}{}}$ since the periods outside $\circtop{\choice{}{}}$ cannot be distinguished by the bird's eye view.

The direct example that satisfies this criterion is a real-valued payoff function $\mathscr{U}_d^\tau:\mathbf{H}^\tau\rightarrow R$ defined as
\[
\mathscr{U}_\ell^\tau(\mathbf{h})
\ =\  \lim_{k\in\FN}\mon{}{
  \begin{matrix}
    \sum_{0\leq i\leq \min(k,1)}\frac{1}{2^{k}}U(\mathbf{h}(t_{(i,0)}))\hspace{3cm}\\
    \hfill+\sum_{2\leq i\leq k}\sum_{2\leq j\leq 2^{i-1}}\frac{1}{2^{k-1}}U(\mathbf{h}(t_{(i,j-2)}))
  \end{matrix}},
  \]
where $U(h)\in\FQ$ for all $h\in Z$.
This function is called the \textit{limit of means payoff function} of the game $\mathbf{\Gamma}^\tau$.

Under this criterion, weak separability is also supported by the preference relations corresponding to these functions.
However, unlike the discounted/simple sum criterion, huge transitivity is not satisfied.

\begin{lemma}\label{lmws}
  A preference relation $\succsim^\tau$ that follows the limit of means criterion satisfies weak separability but does not huge transitivity.
\end{lemma}

\begin{proof}
  For simplicity, it is assumed that there exists $\varepsilon\in \tau\setminus\FN$ such that $\tau=2^\varepsilon$.
  Since $U(h)\in\FQ$ is finite,
  the value of $\mathscr{U}_\ell^\tau(\mathbf{h})$ and hence $(\succsim^\tau)$ is not affected if only one component of the whole history changes.
  This implies that weak separability is trivially preserved.

  Second, let $(\mathbf{h}_k)_{k\in\{1,\ldots,\mu\}}$ for some $\mu\in\varepsilon\setminus\FN$ be a sequence of whole histories, consisting of two component histories $h,j\in H$ whose payoffs are given by $U(h)=1$ and $U(j)=2$, with
  \[
  \mathbf{h}_k \ =\
  \left\{
  \langle\ell,h\rangle\,;\,
  \begin{matrix}
    \left(\forall i\leq k\right)(\forall j\in 2^{i})\\
    \left(\ell\notin\mon{\circtop{=}}{t_{(i,j)}}\right)
  \end{matrix}
  \right\}
  \cup
  \left\{
  \langle\ell,j\rangle\,;\,
  \begin{matrix}
    \left(\exists i\leq k\right)(\exists j\in 2^{i})\\
    \left(\ell\in\mon{\circtop{=}}{t_{(i,j)}}\right)
  \end{matrix}
  \right\},
  \]
  where $t_{(i,j)}$ for $i\in\FN$ and $j\in j(i)$ in which $j(0)=j(1)=1$ and  $j(i)= 2^{i-2}$ is defined as
  \[
  t_{(i,j)}=
  \begin{cases}
    \tau^i & \text{ if } i<2\\
    \frac{2j+1}{2^{i-1}}\cdot\tau & \text{ otherwise}.
  \end{cases}
  \]
  Then, 
  $\mathscr{U}_{\ell}^\tau(\mathbf{h}_k)= \mathscr{U}_{\ell}^\tau(\mathbf{h}_{k+1})$ holds for all $k\in\{1,\ldots,\mu-1\}$.
  However, $\mathscr{U}_{\ell}^\tau(\mathbf{h}_1)=1$ but $ \mathscr{U}_{\ell}^\tau(\mathbf{h}_\mu)=2$.
\end{proof}

Since preference relations following this criterion generally do not satisfy huge transitivity, the validity of backward induction cannot be guaranteed. 
This implies that the assumption of Proposition \ref{ext} is not satisfied.

However, a similar claim can be made instead.
It states that repeating what is realised as a Nash equilibrium in the constituent game $\Gamma$ turns out to be a subgame perfect equilibrium in $\mathbf{\Gamma}^\tau$.

\begin{prop}\label{lm_Nash}
  Let $\mathbf{\Gamma}^\tau$ be a $\tau$-repeated game with $C=Z$ and $(\succsim_i^\tau)$ follow the limit of means criterion.
  The strategy profile ${s}^{N\tau}$ that plays a Nash equilibrium strategy profile $s^N$ of the constituent game $\Gamma$ $\tau$ times is a subgame perfect equilibrium.
\end{prop}

\begin{proof}
%
  Provided that a whole history $\mathbf{h}$ satisfies $|\mathbf{h}|=\ell$ for some $\ell\in\circtop{\choice{}{}}$.
  Then, the payoff yielded by deviating from $s^N_i$ to the strategy ${s}_i$ which yields player $i=\mathbf{P}^\tau(\mathbf{h})$ the most at each constituent game  is given by
  \begin{eqnarray*}
  \lim_{k\in\FN}\mon{}{
    \begin{matrix}
      \frac{1}{2^{k-1}}\cdot U_i(\mathbf{h}(\ell)^\frown O_{\mathbf{h}(\ell)}({s}_{i}|_{\mathbf{h}(\ell)},{s}_{-i}^{N}|_{\mathbf{h}(\ell)}))\hspace{2cm}\\
      \hfill\ +\ \frac{g(k)-1}{2^{k-1}}\cdot U_i(O({s}_{i},s_{-i}^{N}))
    \end{matrix}
    }\\
  &\hspace{-9.2cm} = \ \mon{}{\frac{\tau-\ell}{\tau}\cdot U_i(O(s_i,s_{-i}^{N}))}
   = \ \mon{}{\frac{\tau-\ell}{\tau}}\cdot U_i(O(s_i,s_{-i}^{N}))
  \end{eqnarray*}
  where $g(k)=\left|\left\{j\in 2^{k-1}; \ell <\tau\cdot j/{2^{k-1}}\leq \tau\right\}\right|$.
  Since $s^{N}$ is a Nash equilibrium of the game $\Gamma$, $U_i(s_{i}^N)\geq U_i(s_i,s_{-i}^N)$ holds for all $i\in I$.
  This implies that $s^{N\tau}$ is a subgame perfect equilibrium.
\end{proof}

The assumptions of Proposition \ref{lm_Nash} can be relaxed to some extent.
The following proposition illustrates this fact.

\begin{prop}\label{limit_means}
  Let $\mathbf{s}^*$ be a subgame perfect equilibrium of $\mathbf{\Gamma}^\tau$ in which $C=Z$ holds and $(\precsim_i^\tau)$ follow the limit of means criterion.
  Then, any strategy $\mathbf{s}^\ell$ that obeys $\mathbf{s}^*$ but switches to any strategy $s$ only if $t$ is in a finite subclass of $\circtop{\textrm{Ch}}$ is a subgame perfect equilibrium.
\end{prop}

\begin{proof}
  Let $X$ denote a finite subclass of $\circtop{\text{Ch}}$ in which $\mathbf{s}^\ell$ plays $s$ other than $s^*$. 
  Then, a limit means of the payoff yielded by $\mathbf{s}^\ell$ is equal to that of $\mathbf{s}^{*}$, since
  \begin{eqnarray*}
  \textstyle
  \lim_{k\in\FN}\mon{}{\frac{\sum_{t\in \circtop{\choice{}{}}}(U_i(\mathbf{O}(\mathbf{s}^*(t))) - {U}_i(\mathbf{O}^\tau(\mathbf{s}^\ell)(t)))}{2^{k-1}}}\\
  & \hspace{-4cm}\ =\ \lim_{k\in\FN}\mon{}{\frac{|X|\cdot(U_i(O(s^*)) - U_i(O(s)))}{2^{k-1}}}\ =\ 0
  \end{eqnarray*}
  holds.
\end{proof}

\begin{ex}[Chain Store Game continued from Example \ref{cs}]\label{cs_lm}
  By looking at the game from the bird's eye view, new equilibria can be found.
  An example is given as follows:
  \begin{enumerate}[1)]
  \item All local stores do not establish second stores, or play ``out'',
  \item The chain store reacts aggressively, or plays ``A'', for all $t\in\{1,\ldots,\tau\}$.
  \end{enumerate}
  Since (out, A) is a Nash equilibrium of the constituent game $\Gamma$, the above set of strategies adds up to a new kind of subgame perfect equilibrium of $\mathbf{\Gamma}^\tau$ by Proposition \ref{lm_Nash}.

  It is also intuitively obvious.
  Since any deviation from the set of strategies gives the chain store only an imperceptible amount of payoff, there is no  incentive to deviate.
  Moreover, all local stores only reduce their payoffs by changing their strategies, because by opening second stores they face an aggressive response from the chain store.
  These arguments confirm that the set of strategies is a subgame perfect equilibrium.
\end{ex}

Similar arguments to Proposition \ref{lm_Nash} hold for the mixed extension $\Delta(G)$ of $G$, by replacing mixed strategies with corresponding strategies of the $\tau$-repeated game of $G$.
To illustrate this in more detail, let us extend the framework to allow players to take nondeterministic actions.

Given that $G=\langle I, H, (\succsim_i)\rangle$ is a strategic game, $\Delta(A_i)$ denotes the class of lotteries on $A_i$ and each probability distribution $\sigma_i:\Delta(A_i)\rightarrow R^{|A_i|}$ is called a \textit{mixed strategy} of player $i$.
Then, the \textit{mixed extension of the game $G$} is defined as follows.

\begin{defn}
  The \textit{mixed extension} of the strategic game $G=\langle I,(A_i),(\succsim_i)\rangle$ is the strategic game $\Delta(G)=\langle I,(\Delta(A_i)),(\Delta(\succsim_i)))\rangle$  where $\Delta(A_i)$ is the class of lotteries on $A_i$ and $\Delta(\succsim_i)$ is a preference relation on $\times_{i\in I}\Delta(A_i)$ which is assumed to satisfy the assumptions of von Neumann and Morgenstern.
\end{defn}

Mixed strategy Nash equilibrium is also given as follows.

\begin{defn}[Definition 32.3 of \cite{OR}]
  A mixed strategy Nash equilibrium of a strategic game $G$ is a Nash equilibrium of its mixed extension $\Delta(G)$.
\end{defn}

Next, let us recreate a kind of mixed extension within a framework of generalised repeated games.
Let $(a_{i,1},\ldots,a_{i,|A_i|})$ be an arbitrarily ordered sequence of elements of $A_i$,
and $m$ denote the least common multiple of all denominators of $\sigma_i(a_{i,j})$ for all $a_{i,j}\in A_i$ and $i\in I$.

A whole history $\mathbf{h}^{m^{|I|}}=(h_t)_{t\in m^{|I|}}$ is called a \textit{mixed unit} of $(\sigma_i)_{i\in I}$ of a mixed extension $\Delta(G)$ of $G$ if and only if it satisfies
\[
\mathbf{h}^{m^{|I|}}(t)(i) \ =\ a_{i,j}
\quad\Leftrightarrow\quad
0 < \frac{t}{m^{{i}}} -{z(t,i)}-{\sum_{k< j}\sigma_i(a_{i,k})}\leq \sigma_i(a_{i,j})
\]
where $z(t,i)=\lfloor\frac{{t-1}}{m^i}\rfloor$ indicates how many times the set of actions that make up $\sigma_i$ has been repeated by player $i$ up to period $t$.

A mixed unit can be interpreted as a representation of a mixed strategy by a generalised extended game,
since the ratio of actions taken by each player in the mixed unit matches the ratio of actions taken by each player in the mixed strategy.
To preserve the independence between the probabilities of the actions taken by each players, they must be repeated $m^{|I|-i}$ times with  $m^{i}$ periods. 

As already mentioned in Section \ref{perspective}, the actions taken at the point in time contained in the same monad are indiscernible and are therefore considered to be the same.
However, the actions that make up each mixed unit are different but are contained in the single monad.
This is because they are thought to be randomly chosen.
This is the key interpretation of this construction.
By putting a unit consisting of several actions into a single monad, the same kind of state as that of the mixed strategy is considered to be recreated.

\begin{prop}
  Let $\Delta(G)$ be a mixed extension of a strategic form game $G$ and $\sigma^*=(\sigma_i^*)_{i\in I}$ be a mixed strategy Nash equilibrium.
  Then, there exists a corresponding strategy $\mathbf{s}^\Delta$ of $\mathbf{G}^\tau$ that satisfies $\frac{\left|\{t\,;\,\mathbf{O}(\mathbf{s}^\Delta)(t)(i)=a_i\}\right|}{\tau} \doteq\sigma_i^*(a_i)$ for all $i\in I$ and is a subgame perfect equilibrium of $\mathbf{G}^\tau$.
\end{prop}

\begin{proof}
  Define a strategy $\mathbf{s}^\Delta$ that assigns $\mathbf{s}^\Delta_i(\mathbf{j})=\mathbf{h}^{m^{|I|}}(t)(i)$ where $t=|\mathbf{j}|-\lfloor \frac{|\mathbf{j}|-1}{m^i}\rfloor\cdot m^i$.
  By definition of $\mathbf{h}^{m^{|I|}}$, the equation $\frac{\left|\{t\,;\,\mathbf{O}(\mathbf{s}^\Delta)(t)(i)=a_i\}\right|}{\tau} \doteq\sigma_i^*(a_i)$ is satisfied for all $i\in I$.
  It implies that the payoff yielded from both strategies coincide.
  Since $\sigma^*$ is a Nash equilibrium of $\Delta(G)$, $\mathbf{s}^\Delta$ is a subgame perfect equilibrium of $\mathbf{G}^\tau$ by Proposition \ref{lm_Nash}.
\end{proof}

The strategy $\mathbf{s}^\Delta$ is said to be a \textit{mixed strategy} of $\mathbf{G}^\tau$.

\begin{ex}[Bach or Stravinsky (Example 15.3 of Osborne and Rubinstein\cite{OR})]
  Two people who want to go to a concert have a hard time choosing the program.
  They both want to go out together.
  But they have different musical tastes.
  One prefers traditional music and would like to go to a Bach concert if possible, while the other, who prefers contemporary music, would like to go to a Stravinsky program.

  Player B, who prefers Bach, gains 2 if they go to a Bach concert together, while Player S, who prefers Stravinsky, gains 1.
  Conversely, Player B gains 1 and Player S gains 2 if they go to the Stravinsky concert together.
  They gain 0 if they fail to coordinate to go out together.
  
  The pure strategy Nash equilibria are given as (Bach, Bach) and (Stravinsky, Stravinsky).
  The mixed strategy Nash equilibrium is given as $\Bigl(\left(\frac{2}{3},\frac{1}{3}\right),\ \left(\frac{1}{3},\frac{2}{3}\right)\Bigr)$.
  This equilibrium is reached by the mixed unit ({\footnotesize BB,BB,SB,BS,BS,SS,BS,BS,SS}) and the mixed subgame perfect equilibrium $\mathbf{s}^\Delta$ is given by
  \[
  \mathbf{s}^\Delta_{\text{\footnotesize B}}(\mathbf{h})\ =\
  \begin{cases}
    \text{\footnotesize B} & \text{if } |\mathbf{h}|>\lfloor\frac{|\mathbf{h}|}{3}\rfloor\cdot 3\\
    \text{\footnotesize S} & \text{otherwise}    
  \end{cases}
  \qquad
  \mathbf{s}^\Delta_{\text{\footnotesize S}}(\mathbf{h})\ =\
  \begin{cases}
    \text{\footnotesize B} & \text{if } |\mathbf{h}|>\lfloor\frac{|\mathbf{h}|+5}{9}\rfloor\cdot 9\\
    \text{\footnotesize S} & \text{otherwise}. 
  \end{cases}
  \]
  The player who prefers Bach repeats the action in cycle 3, choosing Bach twice and then Stravinsky once.
  In contrast, the Stravinsky-loving player chooses Bach 3 times and then Stravinsky 6 times in cycle $9$.
  The number of cycle of the Stravinsky-loving player is set to be the square of the Bach-loving player's. 
  This is intended to eliminate dependency between their actions.
  It is also assumed that players cannot alter the length of the cycles themselves since their actions are chosen randomly and they are unaware of the cycles.

  Note that this result consists only of pure strategies.
  The result can be seen as virtually identical to that originally obtained by the mixed Nash equilibrium strategies.
  It can be said that the limit of means criterion expresses a view that sees one-shot strategic games as something that is repeated an enormous number of times, day after day.
\end{ex}

Returning to Proposition \ref{lm_Nash}, this is based on the assumption that all terminal histories are connected. 
The next proposition alters this assumption and provides a new type of equilibrium in such situations.
%

\begin{prop}\label{lm_SPNE}
  Let ${\Gamma}$ be a constituent game with at least two core players whose actions can terminate the game $\mathbf{\Gamma}^\tau$ in any period $t\in \circtop{\choice{}{}}$, or $Z\setminus C\ne\emptyset$.
  It is assumed that $C$ is a singleton and that all core players receive a positive payoff from the history.
  Then, there exists a collection of strategies $\mathbf{s}^{\mathbf{h}}$ which are subgame perfect equilibria and realise an arbitrarily chosen terminal whole history $\mathbf{h}\in\mathbf{Z}^\tau$. 
\end{prop}

\begin{proof}
  To realise $\mathbf{h}$ as an equilibrium path of a subgame perfect equilibrium, construct a strategy $\mathbf{s}^{\mathbf{h}}$ as follows:
  \begin{enumerate}[1)]
  \item each player chooses an action that constitutes $\mathbf{h}$ at the subgame $\mathbf{j}\subseteq \mathbf{h}$ until the period $|\mathbf{h}|$,
  \item every player who can terminate the game terminates at the subgame $\mathbf{j}\supset\mathbf{h}$, while the rest of the players do whatever they want, 
  \item each player $i$ follows $s_i^*$ at the subgame $\mathbf{j}$ which is neither a predecessor, $\mathbf{j}\not\subseteq\mathbf{h}$, nor a successor $\mathbf{j}\not\supseteq\mathbf{h}$, of $\mathbf{h}$.
  \end{enumerate}
  At each subgame extending from the history $\mathbf{h}$, $\mathbf{s}^{\mathbf{h}}$ establishes Nash equilibrium, since all deviations on the equilibrium path by a single player end up in changing the terminal whole histories to the one with at most 1 monad length longer, so that gains by deviation remain indiscernible amount.
  It is also true at each subgame preceeding the history $\mathbf{h}$,
  since any deviation on the equilibrium path leads only to shortenen the path and reduces its payoff.
\end{proof}

\begin{ex}[Centipede Game continued from Example \ref{dc}]\label{dc_lm}
  There are also new equilibria in the centipede games when viewed from the bird's eye view.
  An example is shown below:
  \begin{enumerate}[1)]
  \item Both players play ``R'' or ``r'' until the period $\gamma<\tau$ has come,
  \item and play ``D'' or ``d'' after the period $\gamma$ passed.
  \end{enumerate}
  The set of strategies is found to be a subgame perfect equilibrium of $\mathbf{\Gamma}^\tau$ directly by Theorem \ref{lm_SPNE}.

  It is also intuitively obvious.
  Since any deviation from the set of strategies yields only an indiscernible amount of gain for both players, there is no incentive to deviate.
  This implies that the set of strategies is a subgame perfect equilibrium.
\end{ex}

\section{Concluding Remarks}
The present paper proposes a framework that allows different forms of extended games, which have a repetitive structure in common, to be organised in a unified way, and thereby, specifies the conditions that the newly arising equilibria in these games must satisfy when they are extended beyond the finite horizon.
These new equilibria exist due to the nature of the number system constructed by AST which allows huge or extremely small numbers to be represented in appropriate forms.

One of the most striking features of these equilibria, resulting from the structure of infinity provided by AST, can be seen in the centipede games, exclusively with the payoffs given by the simple sum criterion.
Players in the centipede games behave cooperatively under this criterion, because at some point they have won infinite payoffs, and thus have lost interest in increasing their payoffs further.
These huge payoffs are obtained simply by neglecting the indiscernible amount relative to the total, and the simple sum criterion allows them to behave in this way.

These equilibria are so simple and natural that even players with little knowledge of game theory are observed to play these strategies.\footnote{These aspects are in consistent with many experimental data too. See, for example, a review by Krockow, Coleman and Pulford \cite{review} of the experimental works of centipede games.} 
However, both the discounted sum and the overtaking criteria destroy this sense of satiation, and the new equilibria realised by the simple sum criterion are all eliminated. 
While the discounted sum criterion ignores the distant future, the overtaking criterion cannot afford to overlook even tiny losses.
It is precisely this short-sightedness and greed that prevents the cooperative behavior achieved by the simple sum criterion from materialising under these criteria.\footnote{
These negative results 
remind us of the text by Balzac, the author who portrayed every possible aspect of human behavior, especially those related to lending and borrowing money, in which he summed up the greed of misers:
``Les avares ne croient point \`{a} une vie \`{a} venir, le pr\'{e}sent est tout pour eux'' [Misers have no belief in a future life; the present is their all in all] (Balzac \cite{eugenie}, p.275 [p.128]).
}

The function of the limit of means criterion is also of great interest.
In this criterion, the width of each pair of successive periods is considered to be indiscernibly short, so that the situation is presented as if the whole game were played continuously.
It allows us to model the behaivor of chain stores which ignore each individual action that causes tiny losses, but see the problem as a whole continuum and manage to make enormous profits in the end.

These equilibria are enabled because each core player cannot decide what to do in each period. 
They can only change their behavior after the certain periods have passed which are considered to be discernible from the period they are in by the indiscernibility equivalence $\circtop{=}$.
By introducing topologies in which a pair of indiscernibly close numbers is considered identical, the limit of means criterion allows players to see the whole problem through the bird's-eye view and to make decisions that are not influenced by tiny losses in indiscernibly short periods.
As can easily be seen, this property is also useful when dealing with dynamic problems that change in continuous time.

The present paper also attempts to view the strategic games as continuously repeated games at the same time.
The result shows that the mixed strategy equilibrium can be viewed as a subgame perfect equilibrium of the corresponding continuously repeated games, and can be said to provide a more understandable interpretation of mixed strategy.

By modifying the number system according to AST, the framework presented here allows us to greatly expand the range of social phenomena that can be explained by game theory.
Not only those caused by misers, but also those caused by everyday people, especially in a more natural and intuitive way.


\begin{thebibliography}{10}

\bibitem{eugenie}
Honor\'{e}~de Balzac.
\newblock Eug\'{e}nie {G}randet.
\newblock In {\em \OE vres compl{e}tes de H. de Balzac}, La Com{\'e}die humaine (tome V); premi\`{e}re partie; \'{e}tudes de m\oe urs; deuxi\`{e}me livre. Alexandre Houssiaux, Paris, 1855.
\newblock trans. by Wormeley, Katharine Prescott. Roberts Brothers, Boston, 1889.

\bibitem{review}
Eva~M. Krockow, Andrew~M. Colman, and Briony~D. Pulford.
\newblock Cooperation in repeated interactions: A systematic review of {C}entipede game experiments, 1992-2016.
\newblock {\em European Review of Social Psychology}, 27(1):231--282, 2016.

\bibitem{Kunen}
Kenneth Kunen.
\newblock {\em Set Theory: An Introduction to Independence Proofs}.
\newblock North-Holland, 1980.

\bibitem{OR}
Martin~J. Osborne and Ariel Rubinstein.
\newblock {\em A Course in Game Theory}.
\newblock The MIT Press, Cambridge, USA, 1994.
\newblock electronic edition.

\bibitem{centipede}
Robert Rosenthal.
\newblock Games of perfect information, predatory pricing and the chain-store paradox.
\newblock {\em Journal of Economic Theory}, 25(1):92--100, 1981.

\bibitem{br}
Ariel Rubinstein.
\newblock {\em {Modeling Bounded Rationality}}.
\newblock MIT Press Books. The MIT Press, February 1997.

\bibitem{cogjump}
Kiri {Sakahara} and Takashi {Sato}.
\newblock {An Alternative Set Model of Cognitive Jump}.
\newblock {\em arXiv e-prints}, page arXiv:1904.00613, Apr 2019.

\bibitem{kalina_measure}
Kiri {Sakahara} and Takashi {Sato}.
\newblock A {F}oundation of $\sigma$-superadditive {M}easures -- an note on advancing {K}alina measures --.
\newblock {\em arXiv e-prints}, page arXiv:2303.11636, March 2023.

\bibitem{chainstore}
Reinhard Selten.
\newblock The chain store paradox.
\newblock {\em Theory and Decision}, 9:127--159, 1978.

\bibitem{ast}
Petr Vop\v{e}nka.
\newblock {\em {Mathematics in the {A}lternative {S}et {T}heory}}.
\newblock Teubner Verlagagesellshaft, Leipzig, 1979.

\end{thebibliography}

\end{document}